\documentclass{article}

\usepackage[preprint]{neurips_2025}

\usepackage{color}	
\usepackage{xspace}
\usepackage{graphicx}
\usepackage{color}
\usepackage{enumerate}
\usepackage{hyperref}
\usepackage[font={small,sf}]{caption}
\usepackage{amssymb}
\usepackage{floatflt}
\usepackage{lipsum}
\usepackage{multirow}
\usepackage{soul}
\usepackage[normalem]{ulem}
\sethlcolor{black!10}
\usepackage{mdframed}
\usepackage{amsmath,amsfonts}
\usepackage{comment}
\usepackage{subcaption}
\usepackage{enumitem}
\usepackage{csquotes}
\usepackage{tabularx}
\usepackage{tabularx}
\usepackage{enumitem}
\usepackage{tikz}
\usetikzlibrary{positioning}
\usepackage{mathtools}
\usepackage[a-1b]{pdfx}  
\usepackage[T1]{fontenc}
\usepackage{microtype}
\usepackage{bbm}
\usepackage{fnpct}  
\usepackage{tcolorbox}

\graphicspath{{./figures/}}

\usepackage{algorithm}
\usepackage[noend]{algpseudocode}

\DeclarePairedDelimiterX\set[1]\lbrace\rbrace{\,#1\,}

\usepackage{tikz}
\usetikzlibrary{positioning}

\newcommand*{\Reals}{\mathbb{R}}

\newcommand{\stitle}[1]{\vspace{1mm}\noindent{\textbf{#1}}.}

\newcommand{\eps}{\varepsilon}

\newcommand{\enet}{{$\varepsilon$-{net}}\xspace}

\newtheorem{theorem}{Theorem} 
\newtheorem{lemma}[theorem]{Lemma}
\newtheorem{proposition}[theorem]{Proposition}

\newtheorem{definition}{Definition}

\newtheorem{corollary}{Corollary}

\newenvironment{proof}{\it \paragraph{Proof}}{\hfill$\square$}

\newtcolorbox{highlightbox}{
  colback=blue!10,
  colframe=blue!20,
  arc=0.5mm, 
  fonttitle=\bfseries,
  boxrule=0mm,
  boxsep=0mm,
  left=0mm,
  right=0mm,
  top=0mm,
  bottom=0mm
}

\newtcolorbox{examplebox}{
  colback=blue!10,
  colframe=blue!20,
  arc=2mm, 
  fonttitle=\bfseries,
  boxrule=0mm,
  boxsep=1mm,
  left=0mm,
  right=0mm,
  top=0mm,
  bottom=0mm
}

\newtcolorbox{defbox}{
  colback=orange!10,
  colframe=orange!20,
  arc=2mm, 
  fonttitle=\bfseries,
  boxrule=0mm,
  boxsep=1mm,
  left=0mm,
  right=0mm,
  top=0mm,
  bottom=0mm
}

\newtcolorbox{pbox}{
  colback=black!5,
  colframe=black!30,
  arc=2mm, 
  fonttitle=\bfseries,
  boxrule=0mm,
  boxsep=1mm,
  left=0mm,
  right=0mm,
  top=0mm,
  bottom=0mm
}

\usepackage[utf8]{inputenc} 
\usepackage[T1]{fontenc}    
\usepackage{hyperref}       
\usepackage{url}            
\usepackage{booktabs}       
\usepackage{amsfonts}       
\usepackage{nicefrac}       
\usepackage{microtype}      
\usepackage{xcolor}         
\usepackage{etoc}

\title{HENN: A Hierarchical Epsilon Net Navigation Graph for Approximate Nearest Neighbor Search}

\author{%
  Mohsen Dehghankar\\
  Department of Computer Science\\
  University of Illinois Chicago\\
  Chicago, IL \\
  \texttt{mdehgh2@uic.edu} \\
  \And
  Abolfazl Asudeh \\
  Department of Computer Science \\
  University of Illinois Chicago \\
  Chicago, IL \\
  \texttt{asudeh@uic.edu} \\
}

\begin{document}

\maketitle

\begin{abstract}
Hierarchical graph-based algorithms such as HNSW have achieved state-of-the-art performance for Approximate Nearest Neighbor (ANN) search in practice, yet they often lack theoretical guarantees on query time or recall due to their heavy use of randomized heuristic constructions. Conversely, existing theoretically grounded structures are typically difficult to implement and struggle to scale in real-world scenarios.

We propose the Hierarchical $\varepsilon$-Net Navigation Graph (HENN), a novel graph-based indexing structure for ANN search that combines strong theoretical guarantees with practical efficiency. Built upon the theory of $\varepsilon$-nets, HENN guarantees polylogarithmic worst-case query time while preserving high recall and incurring minimal implementation overhead.

Moreover, we establish a probabilistic polylogarithmic query time bound for HNSW, providing theoretical insight into its empirical success. In contrast to these prior hierarchical methods that may degrade to linear query time under adversarial data, HENN maintains provable performance independent of the input data distribution.

Empirical evaluations demonstrate that HENN achieves faster query time while maintaining competitive recall on diverse data distributions, including adversarial inputs. These results underscore HENN’s effectiveness as a robust and scalable solution for fast and accurate nearest neighbor search.
\end{abstract}

\newcommand{\navgraph}{\mathcal{G}}
\newcommand{\points}{X}
\section{Introduction}
The Approximate Nearest Neighbor (ANN) problem involves retrieving the $k$ closest points to a given query point $q$ in a $d$-dimensional metric space. This problem is foundational in database systems, machine learning, information retrieval, and computer vision, and has seen growing importance in large language models (LLMs)~\cite{jing2024large}, particularly in retrieval-augmented generation (RAG) pipelines, where relevant documents must be retrieved efficiently from large corpora~\cite{fan2024survey, lewis2020retrieval}. 
More generally, any vector database system implements some form of approximate nearest neighbor (ANN) search to enable efficient vector similarity queries~\cite{pan2024survey, han2023comprehensive, aumuller2020ann}.
For a comprehensive overview of additional applications, we refer the reader to the following surveys~\cite{wang2021comprehensive, li2019approximate, liu2021revisiting}.

To address this problem, several classes of algorithms have been developed. Hash-based approaches (e.g., Locality-Sensitive Hashing~\cite{har2012approximate, indyk1998approximate, datar2004locality}) offer theoretical guarantees on retrieval quality but often struggle in practical settings~\cite{pham2022falconn++}. Quantization-based methods cluster data and search among representative centroids, yielding speedups at the cost of approximation error~\cite{jegou2010product, ge2013optimized, ozan2016competitive}. Graph-based approaches~\cite{malkov2018efficient, jayaram2019diskann, wang2021comprehensive}, particularly hierarchical variants~\cite{malkov2018efficient, lu2021hvs, munoz2019hierarchical}, have gained attention due to their strong empirical performance and scalability. These methods build graphs over the dataset and perform greedy traversal to locate approximate neighbors quickly\footnote{Related work is further discussed in Appendix~\ref{sec:related} (supplemental material).}.

Among them, the Hierarchical Navigable Small World (HNSW)~\cite{malkov2018efficient} graph is widely used in practice. HNSW organizes data into multiple layers by assigning each point to a randomly chosen level and constructs navigable small-world graphs at each layer. While HNSW mostly delivers state-of-the-art practical performance and is easy to implement, it lacks formal guarantees on query time, and its worst-case complexity is shown to be \emph{linear to dataset size} in adversarial settings~\cite{indyk2023worst, prokhorenkova2020graph}. This theoretical gap limits its reliability in applications requiring performance bounds.

In contrast, earlier theoretically grounded hierarchical structures, such as Cover Trees~\cite{beygelzimer2006cover} and Navigating Nets~\cite{krauthgamer2004navigating}, provide logarithmic query time guarantees by constructing hierarchies using $r$-nets\footnote{Here, $r$-nets differ from $\eps$-nets as defined in computational geometry. By $r$-net, we refer to a subset of points that ensures every point in the space is within distance $r$ of some net point. In contrast, $\eps$-nets refer to subsets that intersect all "heavy" ranges, containing more than an $\eps$ fraction of the total volume or weight.}. However, these structures are relatively \emph{difficult to implement} and often do not scale well to real datasets, limiting their practical adoption.

This contrast highlights a gap in the literature: {\em no existing method offers the simplicity and scalability of HNSW while also providing provable query time guarantees} of classical hierarchical structures. Theoretical structures with strong guarantees remain underutilized due to their implementation complexity, whereas practical methods often lack worst-case guarantees.

This paper addresses this gap by proposing {\bf Hierarchical $\eps$-Net Navigation Graph (HENN)}, a novel graph-based structure for ANN search (Figure~\ref{fig:henn}).
Like HNSW, HENN constructs a hierarchical layering structure over the dataset, but its hierarchy is formed using $\eps$-nets, a fundamental concept from computational geometry that ensures coverage of all significant neighborhoods. Each layer is built by computing an $\eps$-net of the previous one, resulting in a well-balanced hierarchy that guarantees {\bf polylogarithmic query time in the worst case}, depending on the specific navigation graph used in each layer. These nets can be efficiently constructed via random sampling with minimal overhead, making HENN nearly as simple to implement as HNSW.

\begin{figure}[t]
    \centering
    \includegraphics[width=0.8\linewidth]{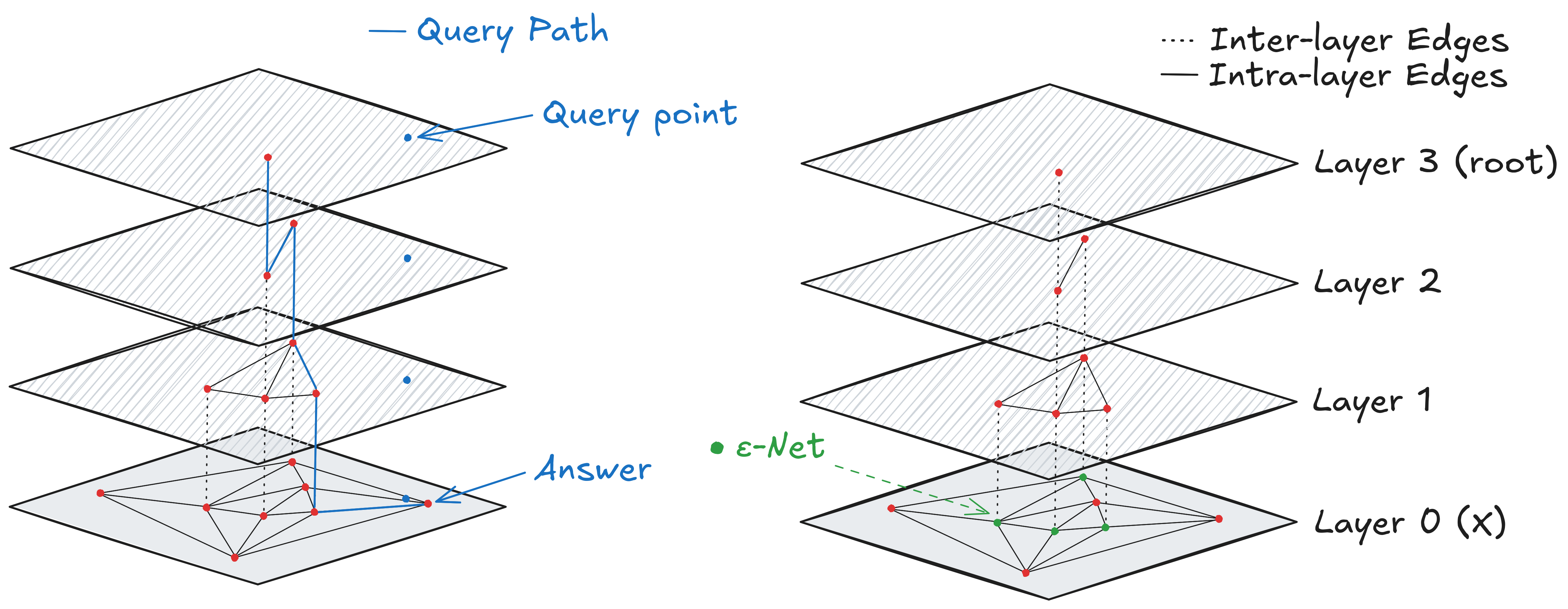}
    \caption{A simple representation of the Hierarchical $\eps$-net Navigation Graph (right) with an example of answering a query using this structure (left). Layers are numbered bottom-up, with layer $0$ being the point set $\points$ and the last layer (layer $L$) called the root.}
    \label{fig:henn}
    \vspace{-7mm}
\end{figure}

Importantly, HENN is \emph{modular}: any graph-based similarity search algorithm can be used within each layer, while the layer construction follows $\eps$-net theory. This flexibility allows HENN to act as a wrapper that enhances existing graph-based solutions with theoretical guarantees depending on the chosen graph.

As the second contribution, we provide a probabilistic analysis of HNSW, showing that its hierarchy implicitly forms an $\eps$-net with high probability, leading to a polylogarithmic query time bound under certain assumptions. This result offers theoretical insight into HNSW’s strong empirical performance.

Our experimental results show that HENN achieves comparable recall and query speed to HNSW on real-world datasets, while significantly outperforming it on adversarial distributions where HNSW’s performance degrades to linear time. Importantly, HENN achieves this without increasing the index size compared to HNSW. However, the improved query speed comes at the cost of increased indexing time due to the repeated sampling required to construct $\eps$-nets. This highlights a trade-off between indexing effort and query-time guarantees, allowing users to choose based on application constraints.

\subsection{Paper Organization}
\vspace{-1mm}
The paper is organized as follows. Section~\ref{sec:henn} introduces the formal notation, defines the problem, and describes the HENN construction during the preprocessing phase, along with an overview of the query algorithm. Section~\ref{sec:theory} presents a formal analysis and proof of the query time guarantees. Section~\ref{sec:exp} reports our experimental results. Section~\ref{sec:discussion} discusses the advantages and limitations of HENN. The Appendix includes extended related work, proofs, pseudo-codes, additional details on constructing $\varepsilon$-nets and navigation graphs, more experimental results, and a discussion of HENN in the dynamic and parallel settings.

\section{Hierarchical $\eps$-net Navigation Graph (HENN)}\label{sec:henn}
\vspace{-2mm}
In this section, after formally introducing the terms and notations, we propose our data structure for answering the approximate nearest neighbor queries.

\stitle{Data Model} We consider a set of $n$ points $\points=\{x_i\}^n$, where $x_i$ is a $d$-dimensional point in $\Reals^d$.
Additionally, we are given a distance function $dist:\Reals^d\times \Reals^d\rightarrow \Reals$ resulting in the metric space $(\points, dist)$\footnote{A distance function $\mathbf{d}$ is defines a metric space if (i) $\mathbf{d}(x,x)=0$, (ii) $\mathbf{d}(x,y)=\mathbf{d}(y,x)$, and (iii) $\mathbf{d}(x,y)+\mathbf{d}(y,z)\geq \mathbf{d}(x,z)$ (triangular inequality)~\cite{harpeled2011geometric}.}, where $dist$ measures the distance (similarity) between every pair of points. In this work, we consider any distance function resulting in a VC-dimensionality of $\Theta(d)$
\footnote{With a slight abuse of the term, we consider the VC-dim of a distance function $dist$ as the VC-dim of range space $(\points, \mathcal{R})$, where $\mathcal{R} = \{dist(x,p)\leq r$, $\forall p\in \Reals^d$ and $r>0\}$.}.
This is the case for most of the distance functions defined in the literature. Specifically, the $\ell_p$-norm between $x_i$ and $x_j$, which is measured as\footnote{Angular distances like cosine similarity also have this property. A discussion on this is also provided in the detailed experiments in Appendix~\ref{sec:app:exp}.}: 
\vspace{-2mm}
\[dist(x_i, x_j) = ||x_i-x_j||_p=\left(\sum_{k=1}^d \left(x_i[k]-x_j[k]\right)^{p}\right)^{\frac{1}{p}}\]

Without loss of generality, we use the Euclidean distance ($\ell_2$-norm) in our examples.

\stitle{Problem setting} Our objective is to preprocess $\points$ and construct a data structure that enables answering the nearest-neighbor queries.
Given a query point $q\in \Reals^d$ and a point set $\points$, a $k$-nearest neighbor ($k$-NN) query aims to find the $k$ closest points $x_i\in \points$ to $q$; i.e., $O_q=k\text{-}\min_{x_i\in X} dist(x_i,q)$.
The Approximate Nearest Neighbor (ANN) is a relaxed version of this problem where the goal is to find a set $A_q\subset \points$ of size $k$ with a high {\em recall rate} -- the probability that each returned point in $A_q$ belongs to $O_q$, formally:
\[
    Recall@k = \frac{|A_q \cap O_q|}{k} \leq 1
\]

We are now ready to define our data structure for ANN query answering:

\begin{definition}[Hierarchical $\eps$-net Navigation Graph (HENN)]
    We define HENN as a multi-layer graph for answering ANN queries on a point set $\points$:
    \begin{itemize}
        \vspace{-2mm}
        \item Each node of the graph represents a point in $\points$.
        \item Each layer $\mathcal{L}_i$, $1 \leq i \leq L$, is an $\eps$-net of the point-set $\points$ and its previous layer, for the range space and the values of $\eps$ specified in Section~\ref{sec:construction}. We use $\mathcal{L}_0$ to show the point set $\points$.
        \item The nodes within each layer are connected with {\em Intra-layer edges} that construct a {\em navigable graph} to answer the ANN {\em only} inside this layer.
        Any of the existing or future navigable-graph techniques, such as the Delaunay Triangulation (DT)~\cite{le1993correlations}, an approximation of the DT graph (like Navigable Small World, NSW~\cite{malkov2014approximate}), or the k-nearest neighbor (kNN) graph, can be applied\footnote{In Appendix~\ref{sec:app:nav-graph} in the supplemental material, we propose heuristics based on ideas such as dimensionality reduction to improve the practical performance of the navigable graph used in HENN.}.
        \item Each pair of layers $\mathcal{L}_i$ and $\mathcal{L}_{i+1}$ ,  $0 \leq i < L$, are connected with {\it Inter-layer edges}. there is an edge between the nodes $v\in \mathcal{L}_i$ and $u\in \mathcal{L}_{i+1}$, if and only if these $v$ and $u$ represent the same point in the point set $\points$.
    \end{itemize}
\end{definition}

Figure~\ref{fig:henn} (right) shows a simple representation of a HENN structure.

Having defined our data structure, we now present its construction details.
\subsection{HENN Graph Construction}\label{sec:construction}
\vspace{-1mm}
In this section, we describe the construction of the HENN graph during the preprocessing phase.
We begin by outlining the high-level structure and then progressively examine each component in greater detail through separate subsections, following a top-down approach.

\stitle{\enet Hierarchy} Each layer $\mathcal{L}_i$ in the HENN graph is an \enet for the range space $(\points, \mathcal{R})$, where $\points$ is the input point set and $\mathcal{R}$ is the family of {\em ring} ranges, defined as follows, that specify subsets of $\points$. For a formal definition of \enet, please refer to the background review provided in Appendix~\ref{sec:app:build-epsnet}.

\begin{definition}[Ring Ranges]
    Given a set of points $\points$, and a distance function $dist$, a {\bf ring} $R\in \mathcal{R}$ is specified with a base point $p\in \Reals^d$ and two values $r_1<r_2$. Any point in $\points$ with distance within two values $r_1$ and $r_2$ from $p$ falls inside the ring. Formally,
    \[
    R\cap \points = \{x\in X~|~ r_1\leq dist(x,p)\leq r_2\}
    \]
\end{definition}

\begin{proposition}\label{prop:1}
    The VC-dim of the ring range space, $(\points, \mathcal{R})$, is $\Theta(d)$.
\end{proposition}

The correctness of Proposition~\ref{prop:1} follows the fact that each ring range $R: \langle p,r_1,r_2\rangle$ can be formulated by mixing the two distance ranges  $R': dist(x,p)\leq r_2$ and $R'': dist(x,p)\leq r_1$ as $R=R'-R''$. Hence, due to the mixing property of range spaces~\cite{harpeled2011geometric}, the VC-dim of the ring ranges is two times the VC-dim of the distance ranges, i.e., $\Theta(d)$. Remember, we assumed that the VC-dim of $dist$ is $\Theta(d)$.

Following Proposition~\ref{prop:1}, one can build an $\eps$-net of size $O(\frac{d}{\eps}\log \frac{d}{\eps})$, for the ring range space, for a given value $\eps$~\cite{har2012approximate} (See $\eps$-net construction detail background in Appendix~\ref{sec:app:build-epsnet}).

\stitle{The construction algorithm} The preprocessing phase follows a recursive construction of the HENN graph using     Algorithm~\ref{alg:preprocess}. It begins with the initial set of points being the entire point set, i.e., $\mathcal{L}_0=\points$, and constructs an $\eps(n)$-net over $\points$, as explained in Appendix~\ref{sec:app:build-epsnet}, where $n = |\points|$ and $\eps(n)$ is defined later in Equation~\ref{eq:eps}. This forms the first layer, denoted $\mathcal{L}_1$. After finding the points in this layer, we follow a black-box approach for constructing a {\it navigable graph} within this layer and add the intra-layer edges accordingly (see Appendix~\ref{sec:app:nav-graph} for more details). Even though this construction follows a sequential order in building layers, a discussion on how to parallelize this step is provided in Appendix~\ref{sec:app:par}.

\begin{algorithm}[h]
\caption{HENN Graph Construction (Preprocess) Algorithm}
\label{alg:preprocess}
\begin{algorithmic}[1]
\Require The set of points $\points$, maximum number of layer $L$, and the exponential decay $m$.
\Ensure The HENN graph $\mathcal{H}$.
\Function{BuildHENN}{$\points, L, m$}
    \State $\mathcal{L}_0 \gets \points$
    \For{$i \leq L$}
        \State $s \gets |\mathcal{L}_{i-1}|$
        \State $\eps \gets $ calculate $\eps(s)$ based on Equation~\ref{eq:eps}
        \State $\mathcal{L}_i \gets \textit{BuildEpsNet}(\mathcal{L}_{i-1}, \eps)$\Comment{build an $\eps$-net, see Appendix~\ref{sec:app:build-epsnet}.}
        \State Connect each node in $\mathcal{L}_i$ to the previous layer (Inter-layer edges).
        \State Build a navigable graph on $\mathcal{L}_i$ (Intra-layer edges).\Comment{See Appendix~\ref{sec:app:nav-graph}}.
    \EndFor
    \State {\bf Return} $\mathcal{H} = \{\mathcal{L}_0, \mathcal{L}_1,\cdots,\mathcal{L}_L\}$ and the edges.
\EndFunction
\end{algorithmic}
\end{algorithm}

Subsequently, the algorithm recursively builds each layer $\mathcal{L}_{i+1}$ as an $\eps(|\mathcal{L}_i|)$-net of the previous layer $\mathcal{L}_i$ and adds the inter-layer edges. This process continues until a total of $L$ layers are constructed, where $L$ is a hyperparameter specifying the depth of the HENN graph. 
The parameter $\eps$ is defined as
\vspace{-2mm}
\begin{align}\label{eq:eps} 
    \eps(s) = c_0 d\frac{\log s}{s}2^{m}, 
\end{align}

where $c_0$ is a sufficiently large constant (to be specified later), $d$ is the dimension of $\points$, and $m$ is a hyperparameter controlling the \underline{exponential decay} in the size of the layers across the hierarchy. As the following lemma shows, increasing $m$ accelerates the exponential decay in the layer sizes.

\begin{lemma}\label{lem:layer-size}
    In the construction described above, by using the equation ~\ref{eq:eps}, the size of each layer $\mathcal{L}_i$ would decrease exponentially. Formally, for every $1 \leq i \leq L$ ($\mathcal{L}_0 = \points$),
    \vspace{-2mm}
    \[
        |\mathcal{L}_i| \leq \frac{|\mathcal{L}_{i-1}|}{2^m}
    \]
\end{lemma}

\begin{proof}
    Proof is provided in Appendix~\ref{sec:app:proofs}.
\end{proof}

{\bf Note:} The value of $m$ (layers exponential decay), in HNSW, is generally selected between $1$ and $6$, based on the tradeoff between the recall and indexing time~\cite{malkov2018efficient}.

\begin{corollary}\label{cor:layer-size}
    The property of being an \enet is preserved under the addition of extra points. Therefore, without loss of generality, we may assume that for all $i$, the sizes of the layers satisfy:
    \vspace{-1mm}
    \[
        |\mathcal{L}_i| = \frac{|\mathcal{L}_{i-1}|}{2^m}
    \]
\end{corollary}

\begin{corollary}\label{cor:layer-count} The largest possible number of layers, $L$, in the HENN graph grows logarithmically with the size of the point set $\points$, where $|\points| = n$. Specifically, we have
\vspace{-1mm}
\[
    L = O\left(\frac{\log n}{m}\right).
\]

\end{corollary}

After explaining the construction of the HENN graph during the preprocessing time, we are now ready to describe ANN query answering.

\subsection{Query Answering Using HENN} 
\vspace{-1mm}
During the query time, given a query point $q\in \Reals^d$, the goal is to find the approximate nearest neighbor of $q$ within $\points$. We follow the Greedy Search algorithm on top of HENN graph: Starting from a random node in the root (layer $\mathcal{L}_L$), we find the nearest neighbor of $q$ within this layer by following a simple Greedy Search algorithm. 
After finding the nearest neighbor $v_i$ in layer $\mathcal{L}_i$, we continue this process, starting from $v_i$ in layer $\mathcal{L}_{i - 1}$. We continue this until reaching the bottom (layer $\mathcal{L}_0$)\footnote{Note that this is the same greedy algorithm used in the literature.}. A pseudo-code of this algorithm is provided in Algorithm~\ref{alg:query} in the Appendix, and a visual illustration is shown in Figure~\ref{fig:henn} (left).
Answering $k$-Nearest Neighbors for $k > 1$ can be achieved by considering a set of candidates at each step in the greedy algorithm~\cite{malkov2018efficient}.
\section{Theoretical Analysis}\label{sec:theory}
\vspace{-2mm}
In this section, we provide theoretical guarantees on the running time of the HENN graph for answering ANN.
The query time for ANN using HENN partially depends on the specific navigation graph used (Appendix~\ref{sec:app:nav-graph}), as in HNSW.
Recall that a navigable graph connects the nodes in each layer $\mathcal{L} \subseteq \points$ of the HENN graph using the intra-layer edges, and is used by the Greedy algorithm at the query time.

Let \(\mathcal{G}_\mathcal{L}= G\left(\mathcal{L},E_\mathcal{L}\right)\) be a navigable graph built for a layer $\mathcal{L}\subseteq X$, where 
$E_\mathcal{L}$ is the set of intra-layer edges.
Navigating through the graph $\mathcal{G}_\mathcal{L}$, while following the edges $E_\mathcal{L}$ greedily, the Greedy algorithm finds a point  $\overline{p}$ as the approximate nearest neighbor of $q$ among the nodes in $\mathcal{L}$. Formally, 
\[
    \overline{p} = \mathsf{GS}(\mathcal{G}_\mathcal{L}, q)
\]
Where $\mathsf{GS}$ is short for \emph{GreedySearch} in Algorithm~\ref{alg:query}.
In the following, we define a notation of accuracy for a navigable graph algorithm $\mathcal{G}$, which is used as a parameter in the subsequent guarantees that we provide.

\begin{definition}[Recall Bound: $\rho_\delta$]
    For a given Navigable Graph $\mathcal{G}_\mathcal{L}$ on a subset $\mathcal{L}$, the Recall Bound of $\mathcal{G}_\mathcal{L}$ with a probability $\delta$ is defined as the smallest value $k$, such that for all queries $q$, the result of greedy search on $\mathcal{G}_\mathcal{L}$ will give at least one point within the $k$-nearest Neighbors of $q$ with probability more than $\delta$. In other words,
    \vspace{-2mm}
    \[
        \rho_\delta = \arg\min_{k \leq |\mathcal{L}|} \left(Pr\Big(\mathsf{GS}(\mathcal{G}_\mathcal{L}, q) \in \mathsf{NN}_{k, \mathcal{L}}(q)\Big) \geq \delta \right)
    \]
    \vspace{-1mm}
    Where $\mathsf{NN}_{k, \mathcal{L}}(q)$ is the ground-truth $k$-nearest neighbors of $q$ in $\mathcal{L}$. The probability is taken over all possible initial nodes of the graph $\mathcal{G}_\mathcal{L}$, as a starting node.
\end{definition}

\paragraph{Note:} This defines a different and weaker notion of accuracy for graph-based algorithms compared to the standard recall rate. While $Recall@k$ measures the fraction of the true $k$ nearest neighbors retrieved, here we are interested in identifying the smallest value of $k$ such that, with high probability, at least one of the true $k$ nearest neighbors is returned by the algorithm.

In practice, we are usually concerned with a fixed probability $\delta \geq 0.9$. 
In addition, for most of the existing navigable graphs~\cite{malkov2014approximate, malkov2018efficient}, $\rho_{0.9} = O(1)$. We will compare these values by running additional experiments in the Appendix.
\vspace{-2mm}
\subsection{Running Time}
\vspace{-1mm}
We start by proving a lemma that focuses only on a \emph{single pair of layers} (the bottom-most layer and the one above it) and shows the usefulness of using $\eps$-nets in this hierarchy.

\begin{lemma}\label{lem:two-layer-case}
    Let $\mathcal{L}$ be an $\eps$-net of $\points$, where $n = |\points|$. Let $\overline{p} = \mathsf{GS}(\mathcal{G}_\mathcal{L}, q)$ for a query point $q$. Then, the number of points in $\points$ closer than $\overline{p}$ to $q$ is less than $\eps \cdot (\rho_\delta + 2) \cdot n$, with probability more than $\delta$. Formally, with probability more than $\delta$, we have:
    \vspace{-0.5mm}
    \begin{align}
        |\mathsf{NN}_{\leq dist(\overline{p}, q)}|
         \leq \eps \cdot (\rho_\delta + 2) \cdot n
    \end{align}
    \vspace{-1mm}
    where $\mathsf{NN}_{\leq dist(\overline{p}, q)} = \{p\in X \mid dist(p, q) \leq dist(\overline{p}, q)\}$
\end{lemma}

\begin{proof}
    Sketch: This is a result of finding $\eps$-nets on \underline{ring} ranges, where it bounds the total number of points around $q$. The proof is provided in Appendix~\ref{sec:app:proofs}.
\end{proof}

Now, we are ready to propose the main theorem on the time complexity of ANN-query answering using HENN graph:

\begin{theorem}\label{thm:time}
    Given a set of points $\points \subset \mathbb{R}^d$ and the hyperparameter $m$ for the exponential decay, let $\mathcal{H}$ be the HENN graph built using the function $\mathsf{BuildHENN}$ (Algorithm~\ref{alg:preprocess}). Then, the running time of $\mathsf{Query}(\mathcal{H}, q)$ (Algorithm~\ref{alg:query}) is 
    $O(d \cdot d^* \cdot \rho_\delta \cdot \log^2 n)$.
    Where $d^*$ is the average degree of a node in the navigation graphs $\mathcal{G}$ and $\rho_\delta$ is the Recall Bound of the navigable graphs.
\end{theorem}

\begin{proof}
    Sketch: This can be proved by applying Lemma~\ref{lem:two-layer-case} inductively, considering layers $\mathcal{L}_1$ up to $\mathcal{L}_L$.
    The proof is provided in Appendix~\ref{sec:app:proofs}
\end{proof}

In most of the navigation graphs used in practice, the degree $d^*$ is fixed as a constant. As a result, the running time is \emph{polylogarithmic in $n$} in these settings.
The running time also depends on the Recall Bound of the navigation graph $\mathcal{G}$ used in the HENN construction. The better navigation graphs in terms of accuracy also result in faster query answering. As long as $\rho_\delta$ is a constant (or $O(\log n)$), the running time will still be \emph{polylogarithmic in $n$}.

\begin{corollary}[HNSW Guarantees]\label{cor:hnsw}
    Since a random sample of a point set forms an $\eps$-net with a certain probability, Theorem~\ref{thm:time} implies that HNSW, which constructs each layer through random sampling, achieves the same polylogarithmic query time, assuming that each layer forms an $\eps$-net with a specific probability. Note that HENN improves on this by explicitly enforcing the $\eps$-net property deterministically (depending on the Recall Bound).
\end{corollary}
\vspace{-4mm}
\subsection{Preprocessing and Indexing}
\vspace{-1mm}
The index size of HENN is exactly similar to HNSW, which is linear in $n$. Following the Corollaries~\ref{cor:layer-size} and \ref{cor:layer-count}, this size can be calculated as:$\sum^L_{i=0} |\mathcal{L}_i| = \sum^L_{i=0} \frac{n}{2^m} = O(n)$.

The preprocessing time for constructing the HENN index depends on the chosen method for building the $\eps$-net. When using random sampling, a Monte Carlo procedure is employed, repeating the sampling a \underline{constant} number of times until an \enet is obtained (see Appendix~\ref{sec:app:build-epsnet} for details).
\vspace{-4mm}
\section{Experiments}\label{sec:exp}
\vspace{-2mm}
In this section, we evaluate the performance of the HENN structure on both synthetic (challenging) data instances and widely used approximate nearest neighbor (ANN) benchmark datasets. This section includes a comparative analysis between HENN and the HNSW structure. Additional experiments, covering a broader range of benchmark datasets and comparisons with alternative navigable graph baselines (beyond NSW), are provided in the supplemental material (Appendix~\ref{sec:app:exp}). 
Our code is built upon the widely-used standard C++ implementation of HNSW, hnswlib\footnote{\href{https://github.com/nmslib/hnswlib}{https://github.com/nmslib/hnswlib}}. The source code is available in \color{blue}\href{https://github.com/UIC-InDeXLab/HENN}{this repository}\color{black}. 

\subsection{Experiments Setup} 
\vspace{-2mm}
The experimental environment configuration is described in Appendix~\ref{sec:app:exp:conf}.

\stitle{Implementation Details} The core step in constructing the HENN index involves computing an $\eps$-net for each layer, beginning with the base layer $\mathcal{L}_0$, which includes the entire dataset. Appendix~\ref{sec:app:exp:impl} provides a detailed explanation of this process, along with the heuristics used in our implementation.

\stitle{Methods and Datasets} 
We compare HENN against the HNSW baseline, highlighting their key difference in layer construction: HENN enforces $\eps$-net layers, while HNSW uses random sampling. Both use navigable graphs built via greedy insertion. Experiments are conducted on both synthetic and real benchmark datasets, including SIFT~\cite{jegou2010product}, GloVe~\cite{pennington2014glove}, and MNIST~\cite{lecun1998gradient}. Additional implementation details, datasets details, and results are provided in the Appendix~\ref{sec:app:exp:method}.

\stitle{Evaluation}
To evaluate the methods, particularly in terms of runtime, we focus on worst-case performance. Specifically, for each reported value, we execute the method multiple times and return the \underline{maximum} runtime observed, rather than the average\footnote{Because in the average case, HNSW will behave like HENN based on Corollary~\ref{cor:hnsw} and Theorem~\ref{thm:time}.}.

\subsection{Experiment Results}
\vspace{-2mm}
We begin by presenting experiments on standard benchmark datasets, demonstrating that HENN achieves \underline{comparable recall and query time} to the baseline methods. In the subsequent subsection, we turn to more challenging synthetic datasets with controlled skewness and non-uniformity, where HENN consistently \underline{outperforms} the baselines.

\subsubsection{Benchmark Datasets}
\vspace{-2mm}
In this subsection, we compare the methods on three real datasets: SIFT, GloVe, and MNIST. Other benchmark comparisons are included in the supplemental material. 

\stitle{Query Speed vs. Recall tradeoff} Figure~\ref{fig:sift-glove} presents the query speed versus recall@10 for different datasets. Query speed is measured in queries per second (QPS), defined as the inverse of query time. Higher values of both recall and speed (toward the top-right of the plot) indicate a more effective ANN algorithm. On the SIFT dataset, HENN achieves performance comparable to the baseline in terms of both speed and recall. On the GloVe and MNIST datasets, HENN attains a higher area under the curve (AUC), indicating improved overall efficiency.

These results show that HENN maintains the effectiveness of prior methods on standard benchmarks, where ANN search is relatively less challenging in terms of worst-case behavior, where the difficulty primarily stems from high dimensionality and large data sizes. This also indicates that the SIFT dataset is fairly uniformly distributed and not skewed enough to challenge the baseline compared to the GloVe and MNIST.

\begin{figure}[ht]
    \centering
      \begin{minipage}[b]{0.32\textwidth}
        \includegraphics[width=\textwidth]{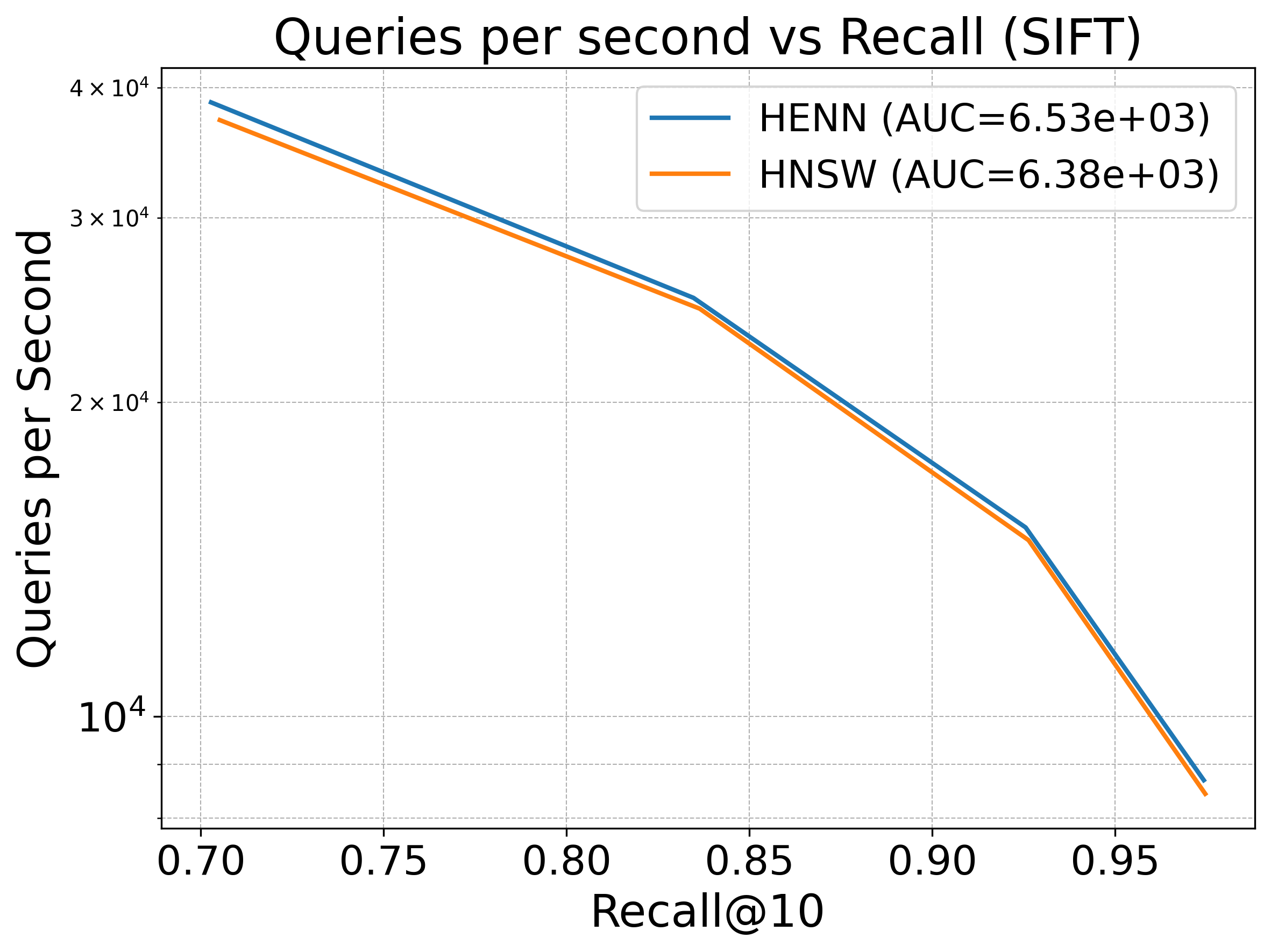}
        \caption*{SIFT: $d=128$}
      \end{minipage}
      \hfill
      \begin{minipage}[b]{0.32\textwidth}
        \includegraphics[width=\textwidth]{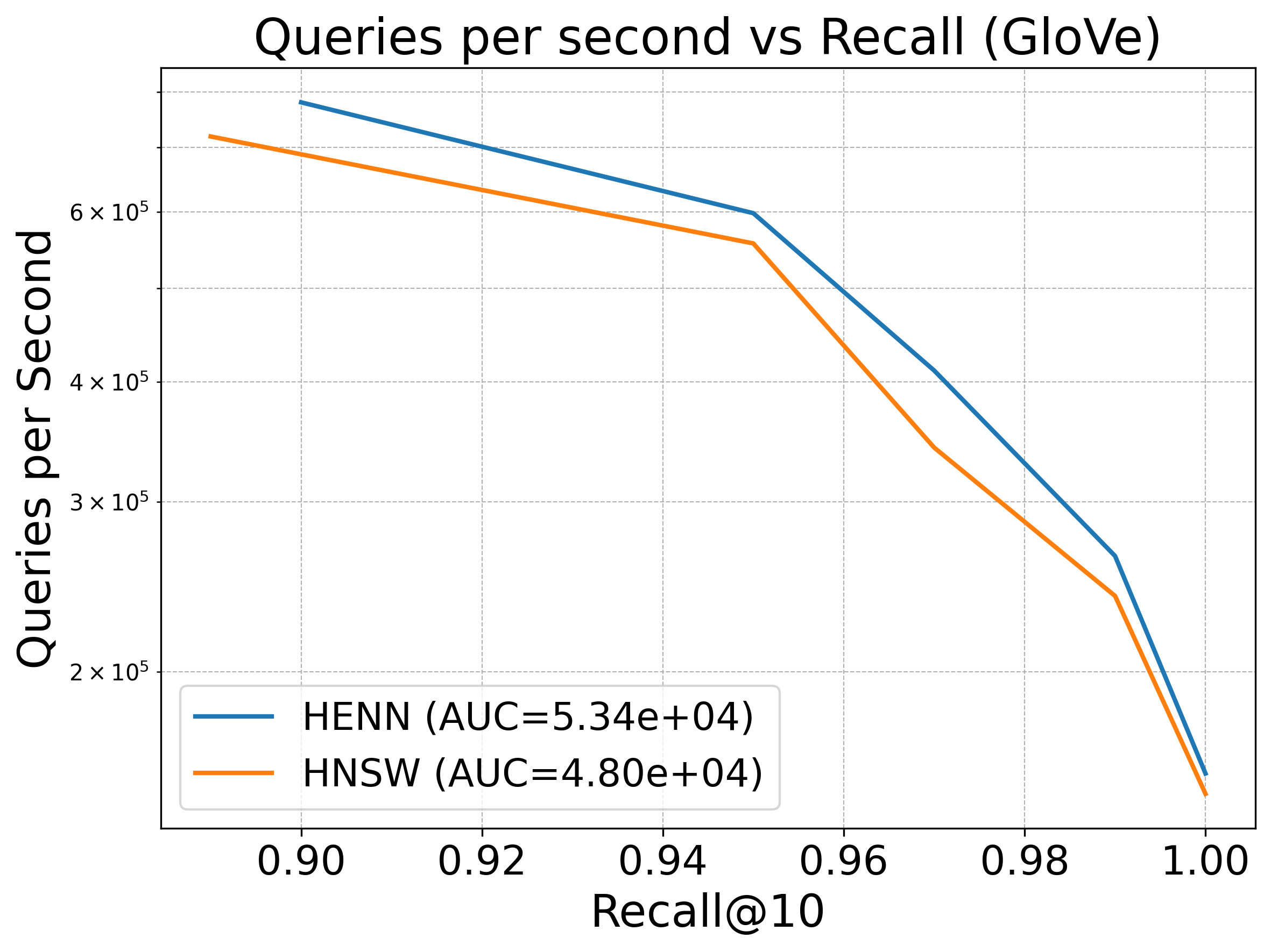}
        \caption*{GloVe: $d=100$}
      \end{minipage}
      \hfill
      \begin{minipage}[b]{0.32\textwidth}
        \includegraphics[width=\textwidth]{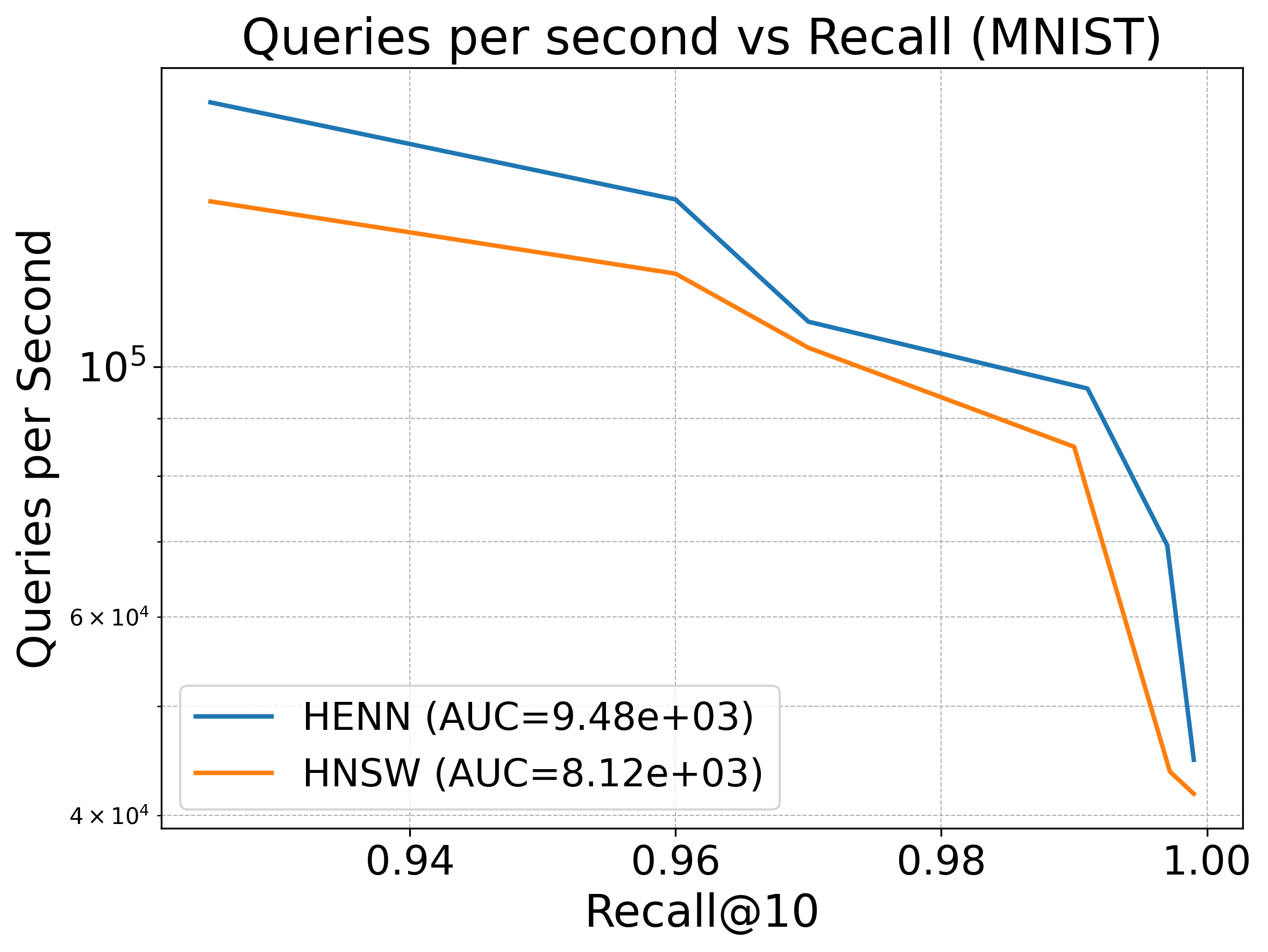}
        \caption*{MNIST: $d=784$}
      \end{minipage}
      \caption{Query speed vs. Recall@10 - up and to the right is better (higher AUC). HENN performs at least as well as the HNSW baseline on standard benchmark datasets. These plots are generated by varying the hyperparameter $ef$ in hnswlib.}
      \label{fig:sift-glove}
      \vspace{-4mm}
\end{figure}

\subsubsection{Synthetic Dataset}
\vspace{-2mm}
As previously discussed, we construct a non-uniform, skewed synthetic dataset with varying parameters to compare the worst-case performance of the methods. 

\stitle{Effect of $n$} Figures~\ref{fig:time-size-16} and \ref{fig:time-size-32} present a runtime comparison between HENN and HNSW as the dataset size $n$ varies. Notably, HENN achieves up to a \underline{1.7$\times$ speedup} over the baseline HNSW, with a larger performance gap as the dataset size increases. Additionally, the runtime of HENN shows a \underline{polylogarithmic growth} trend, consistent with the theoretical guarantee established in Theorem~\ref{thm:time}.

\stitle{Effect of $d$} Figure~\ref{fig:time-dim} presents a runtime comparison as the dimensionality $d$ varies. Similar to previous results, HENN maintains a consistent speedup. Additionally, the runtime of HENN exhibits a linear growth with respect to the dimensionality $d$.

\begin{figure}[ht]
    \centering
    \begin{minipage}[t]{0.32\textwidth}
        \includegraphics[width=0.99\textwidth]{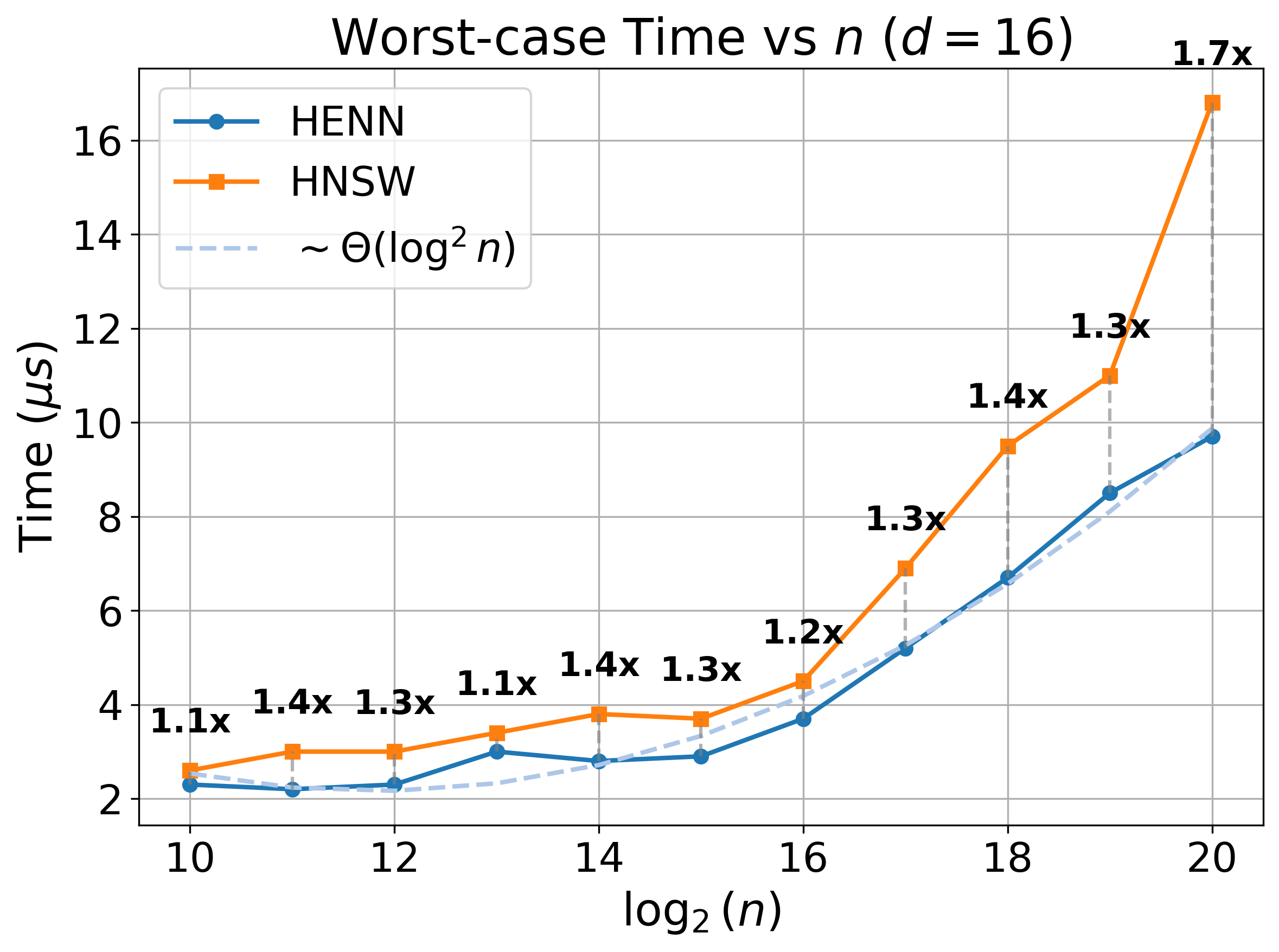}        
        \caption{$d = 16$}
        \label{fig:time-size-16}
    \end{minipage}
    \hfill
    \begin{minipage}[t]{0.32\textwidth}
        \includegraphics[width=0.99\textwidth]{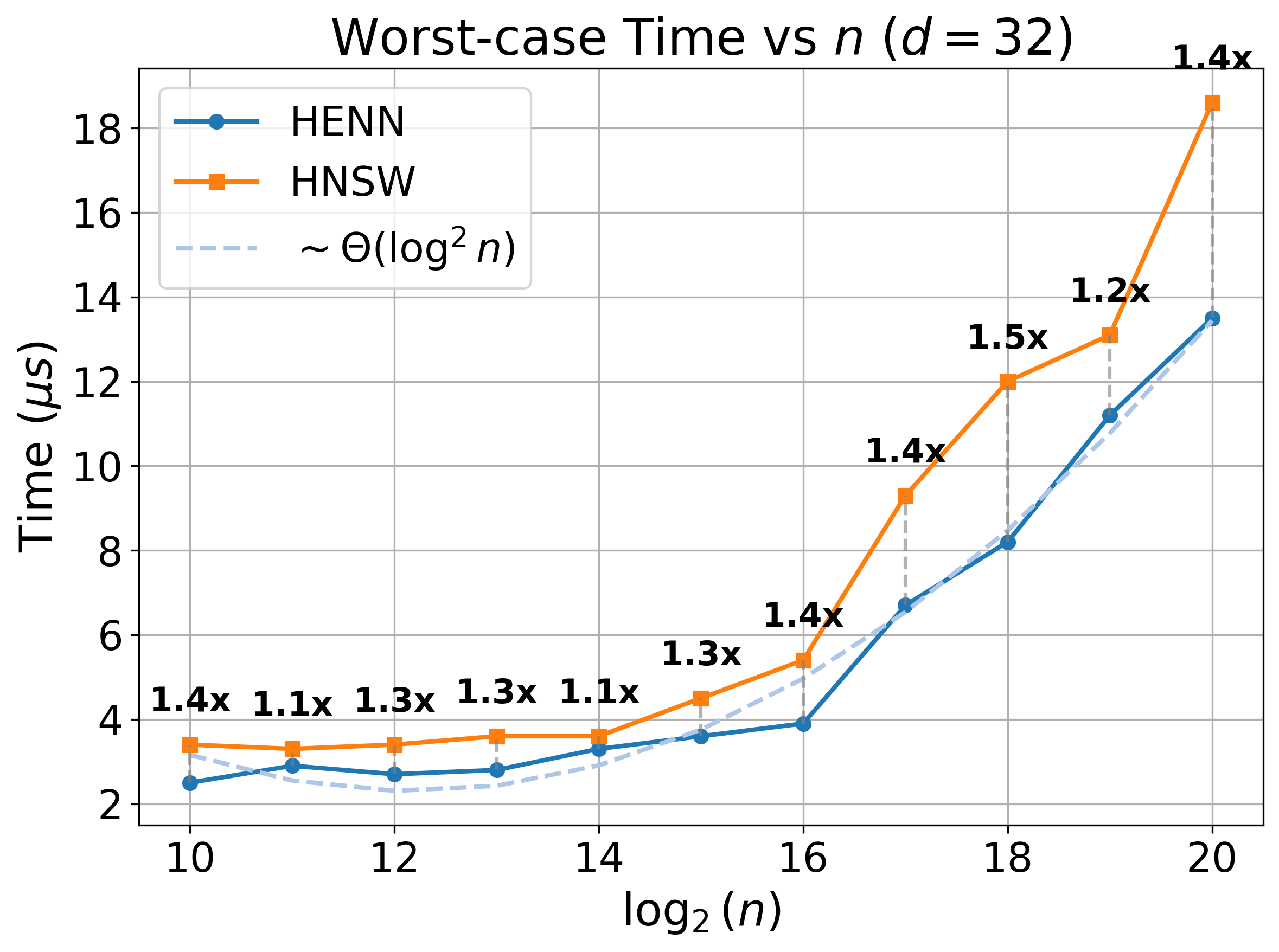}
        \caption{$d = 32$}
        \label{fig:time-size-32}
    \end{minipage}
    \hfill
    \begin{minipage}[t]{0.32\textwidth}
        \includegraphics[width=0.99\textwidth]{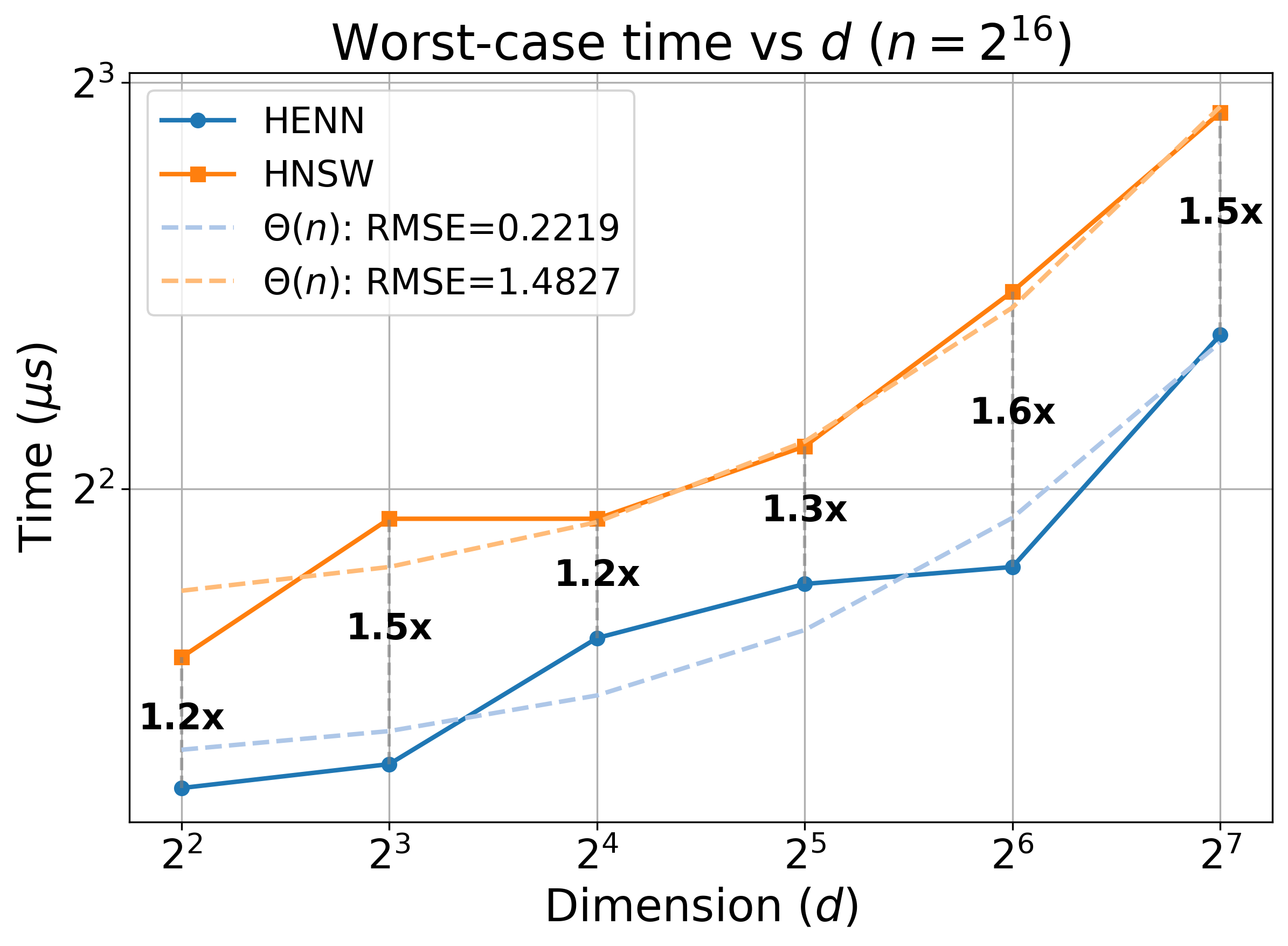}
        \caption{$n = 2^{16}$}
        \label{fig:time-dim}
    \end{minipage}
    \caption*{Comparing the worst-case time of the methods. The worst-case time is calculated by running each method multiple times and using the maximum value as the aggregation.
    The logarithmic trend is shown with a dotted line. RMSE, root mean squared error, measures how well the trends are fitted.}
    \vspace{-4mm}
\end{figure}

\stitle{Number of Hops} Figures~\ref{fig:hops-n-16} and \ref{fig:hops-n-32} compare the two methods in terms of the total number of graph hops visited during the query execution. This metric serves as a proxy for the number of greedy steps performed by Algorithm~\ref{alg:query} to locate the approximate nearest neighbor. HENN consistently required fewer hops, leading to \underline{fewer distance computations} overall. The improvement is as great as \underline{1.9$\times$} fewer hops for HENN. This behavior reflects the more uniform and structured construction of hierarchical layers in HENN compared to the baseline. Figure~\ref{fig:hops-dim} shows this comparison on different dimensions $d$.
\vspace{-3mm}

\begin{figure}[ht]
    \centering
    \begin{minipage}[t]{0.32\textwidth}
        \includegraphics[width=0.99\textwidth]{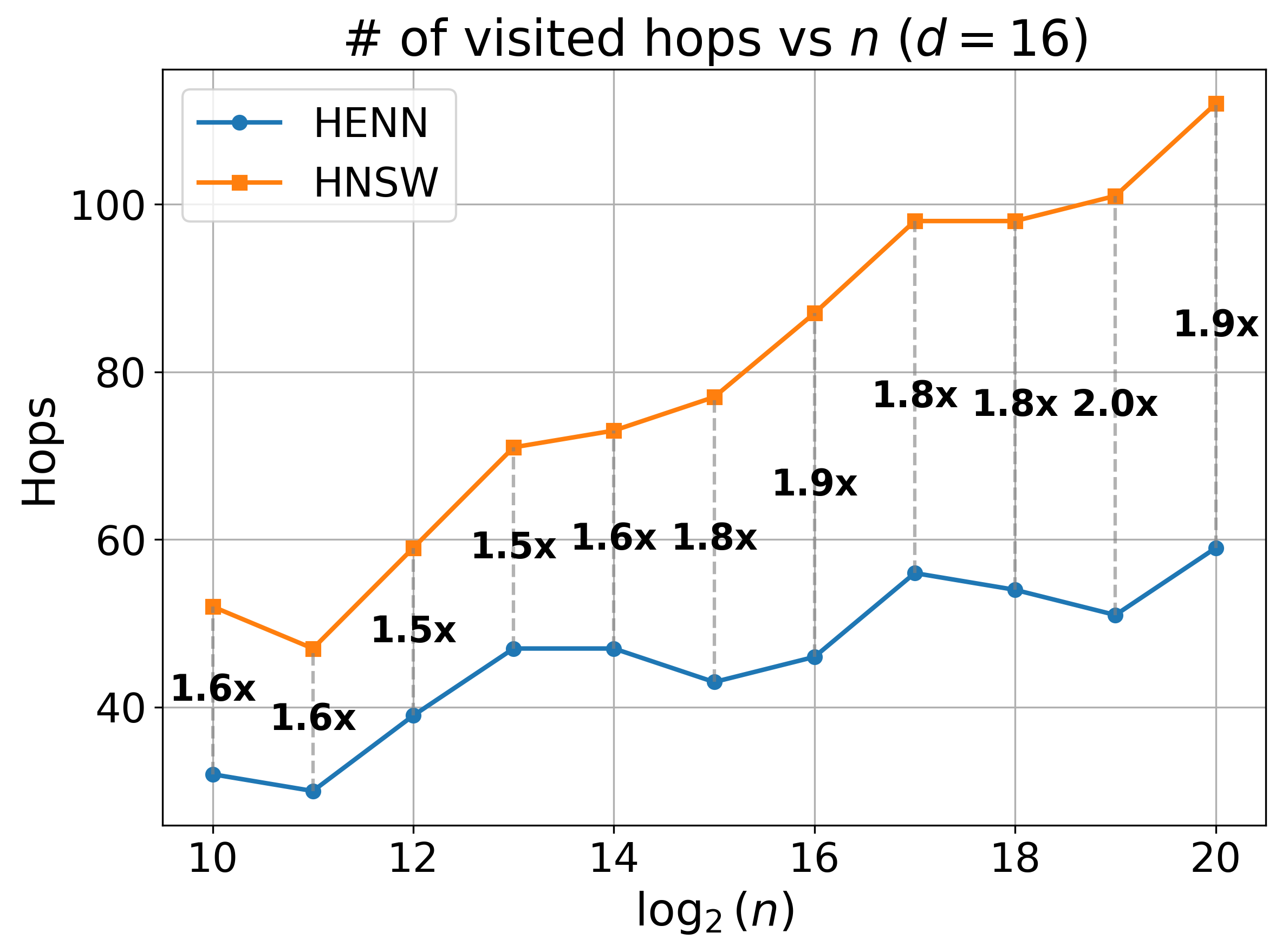}        
        \caption{$d = 16$}
        \label{fig:hops-n-16}
    \end{minipage}
    \hfill
    \begin{minipage}[t]{0.32\textwidth}
        \includegraphics[width=0.99\textwidth]{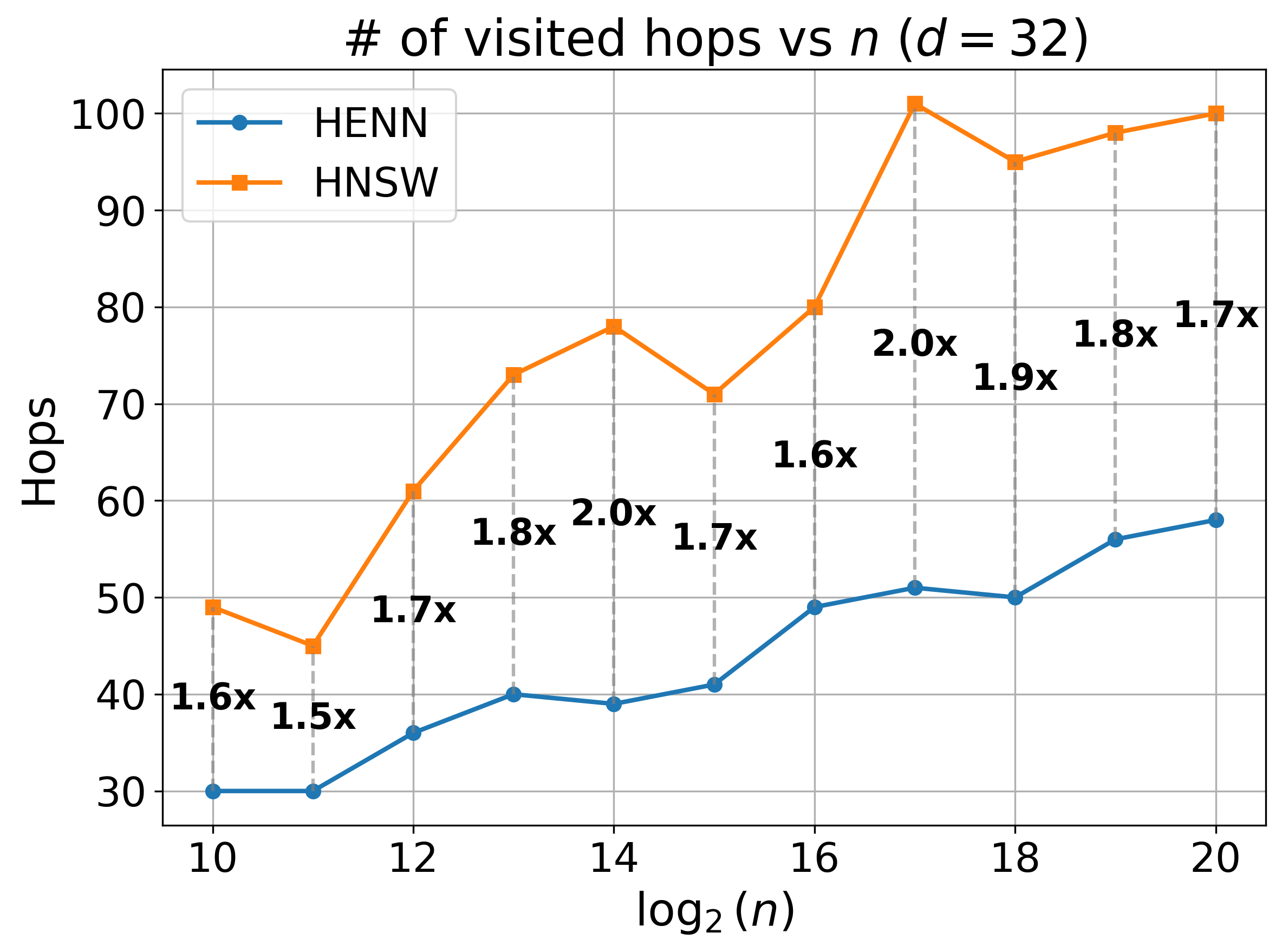}
        \caption{$d = 32$}
        \label{fig:hops-n-32}
    \end{minipage}
    \hfill
    \begin{minipage}[t]{0.32\textwidth}
        \includegraphics[width=0.99\textwidth]{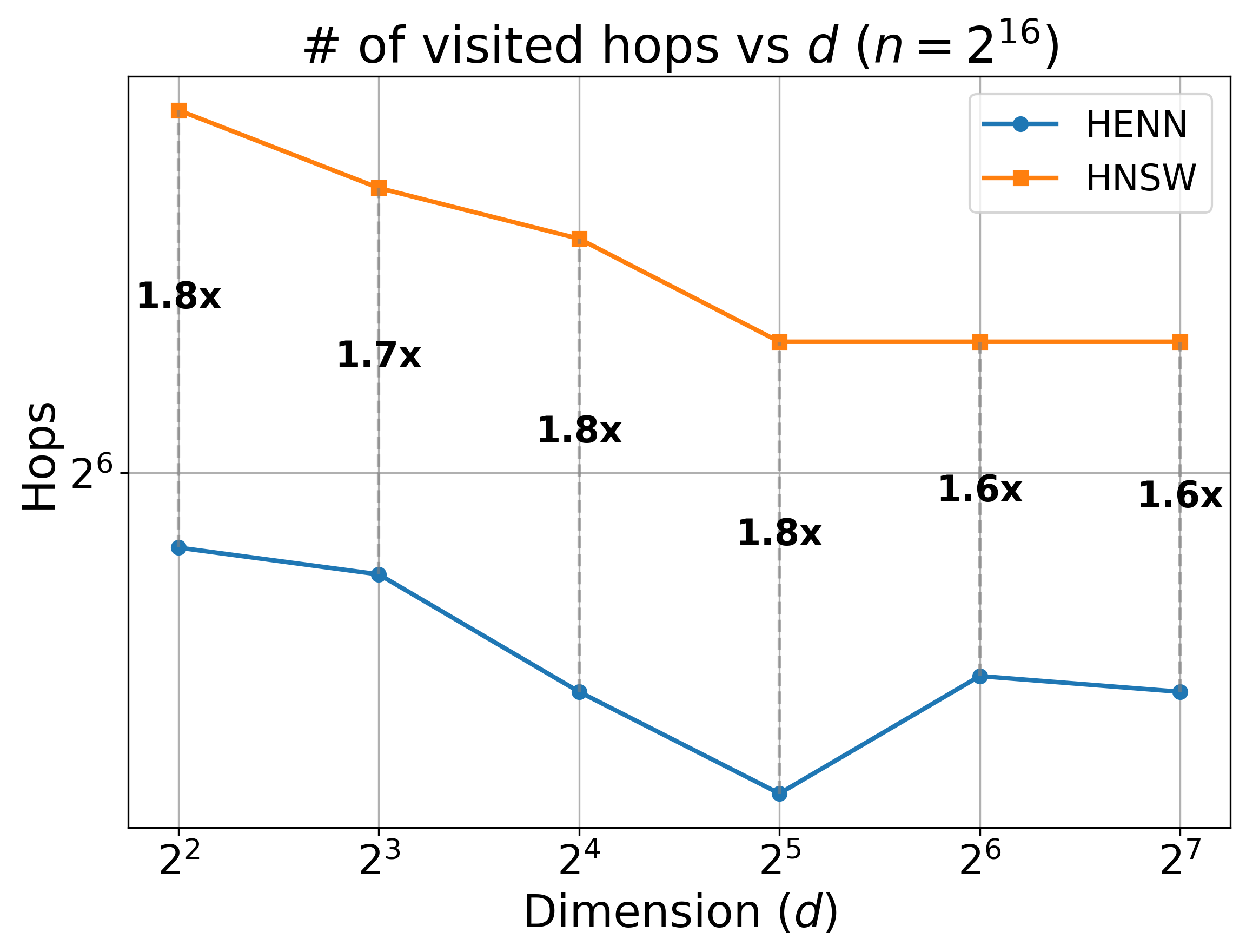}
        \caption{$n = 2^{16}$}
        \label{fig:hops-dim}
    \end{minipage}
    \caption*{Comparing the number of visited hops during the query time. The larger number of visited hops results in longer and less efficient query times.}
    \vspace{-2mm}
\end{figure}

\stitle{Query Speed vs. Recall tradeoff} Figure~\ref{fig:qps-recall} shows the trade-off between query speed and recall@10 for the evaluated methods. We vary the parameter $\lambda$ to control the skewness of the data distribution, with larger $\lambda$ producing more non-uniform datasets. Under highly skewed conditions, HENN demonstrates more speedup over HNSW. In more uniform settings, HENN achieves comparable recall to HNSW while maintaining competitive query speed.

\begin{figure}[htbp]
  \centering
  \begin{minipage}[b]{0.32\textwidth}
    \includegraphics[width=\textwidth]{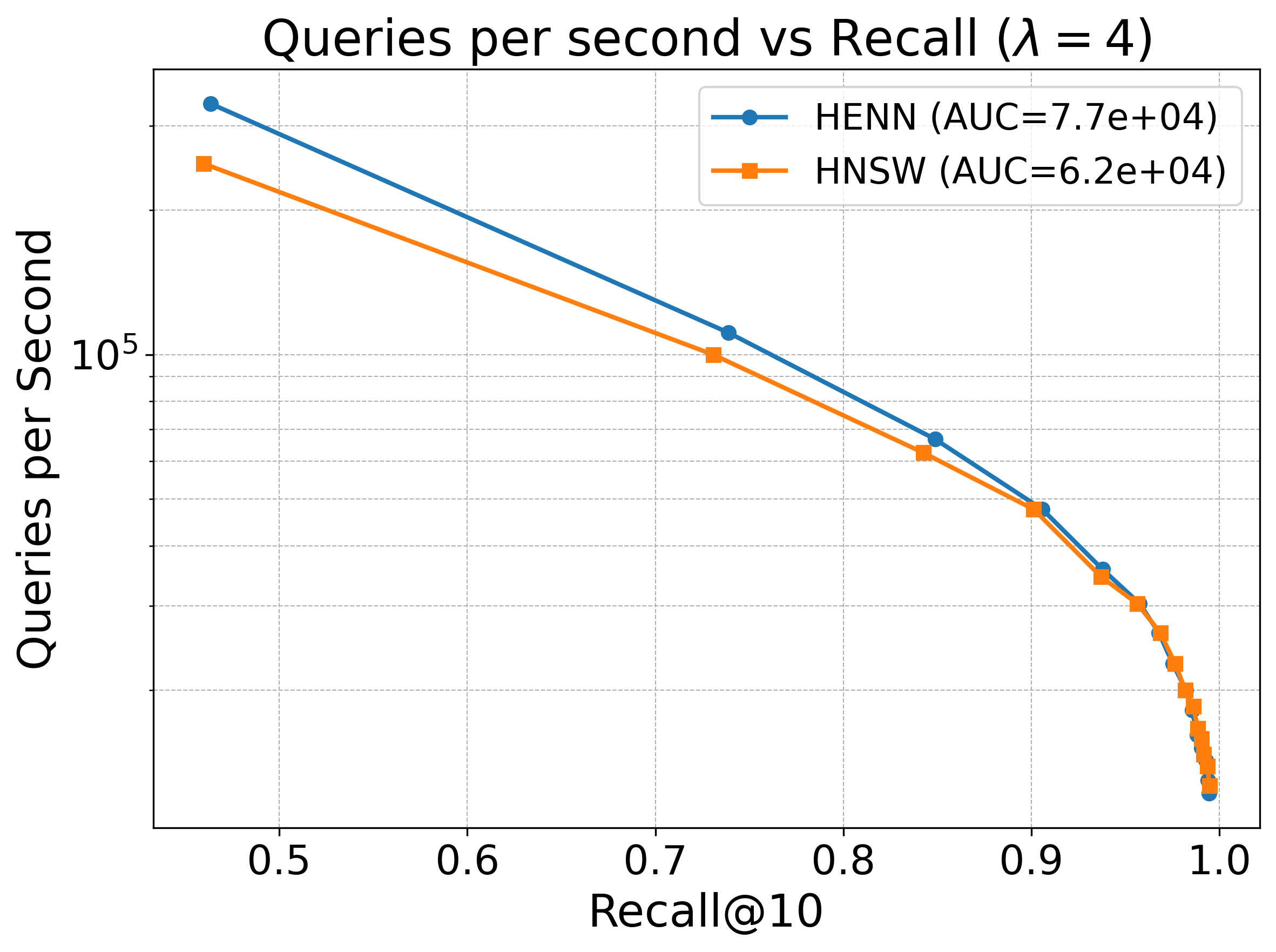}
    \caption*{$\lambda = 4$}
  \end{minipage}
  \hfill
  \begin{minipage}[b]{0.32\textwidth}
    \includegraphics[width=\textwidth]{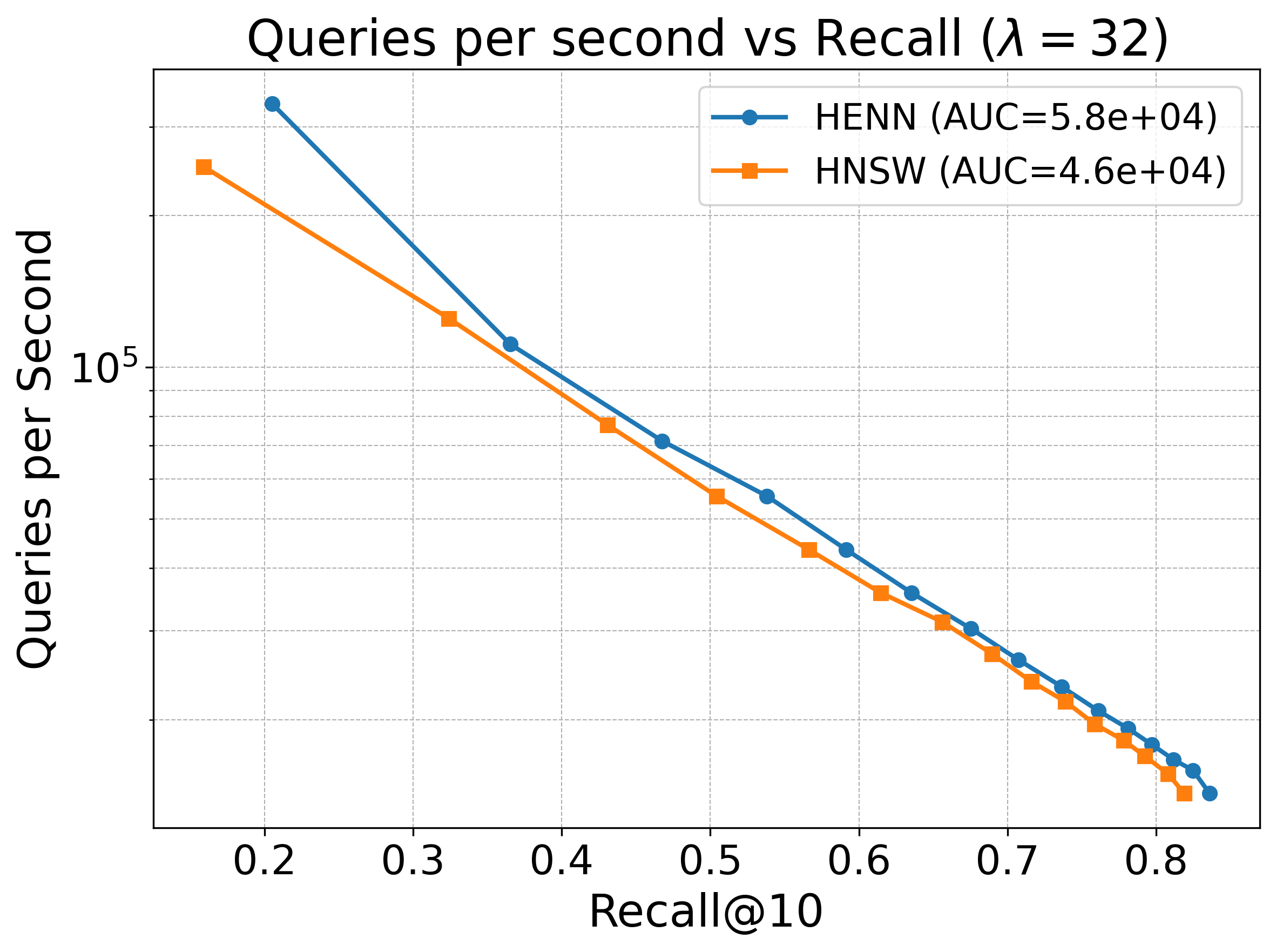}
    \caption*{$\lambda = 32$}
  \end{minipage}
  \hfill
  \begin{minipage}[b]{0.32\textwidth}
    \includegraphics[width=\textwidth]{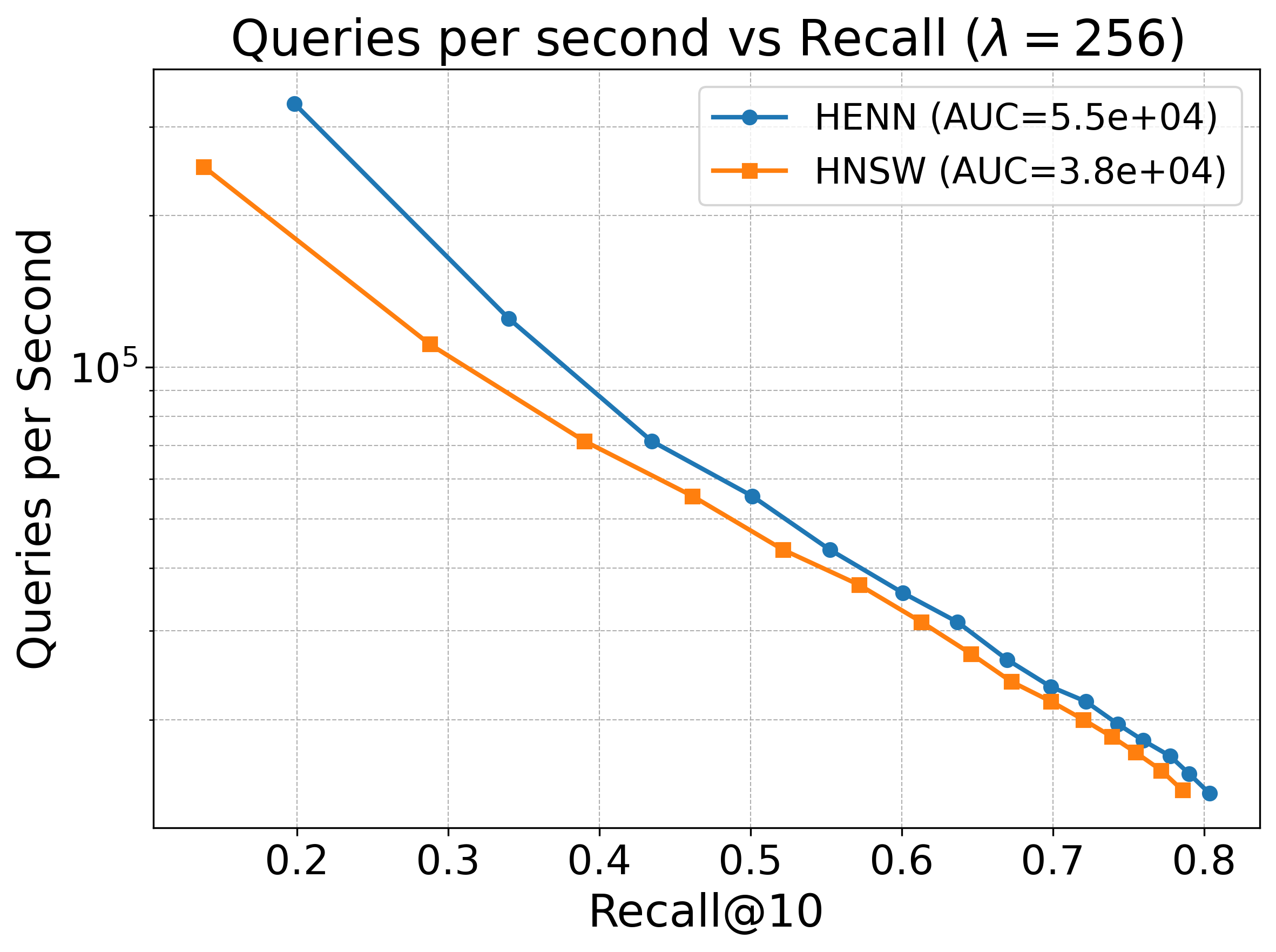}
    \caption*{$\lambda = 256$}
  \end{minipage}
  \caption{Query speed vs Recall@10 - up and to the right is better (higher AUC). The input is an exponential distribution. The larger the $\lambda$, the more skewed the data. In this setting, $d=32$ and $n = 50000$.}
  \label{fig:qps-recall}
  \vspace{-4mm}
\end{figure}

\stitle{Indexing Phase} Table~\ref{tab:index-size} compares the index sizes across datasets. Due to the construction strategy of HENN, there is no significant increase in index size relative to the baseline, as the layer sizes remain unchanged. Indexing time results are shown in Table~\ref{tab:index-time}, where we observe an increase in construction time for HENN. This overhead arises from the repeated random sampling required to satisfy the $\eps$-net condition. A detailed discussion of the trade-off between the success probability of finding an $\eps$-net and indexing time is provided in the Appendix~\ref{sec:app:exp}. 

\begin{table}[h]
\centering
\begin{minipage}[t]{0.48\textwidth}
    \centering
    \begin{tabular}{lccc}
    \toprule
    \textbf{Method} & \textbf{SIFT} & \textbf{GloVe} & \textbf{Synthetic} \\
    \midrule
    HENN  & 633.89 & 26.35 & 18.00 \\
    HNSW  & 633.88 & 26.35 & 18.00 \\
    \bottomrule
    \end{tabular}
    \vspace{2mm}
    \caption{Index size (MB). The Synthetic data has $n=2^{16}$, $d=32$}
    \label{tab:index-size}
\end{minipage}
\hfill
\begin{minipage}[t]{0.48\textwidth}
    \centering
    \begin{tabular}{lccc}
    \toprule
    \textbf{Method} & \textbf{SIFT} & \textbf{GloVe} & \textbf{Synthetic} \\
    \midrule
    HENN  & 79.60 & 1.36 & 1.01 \\
    HNSW  & 20.28 & 0.24 & 0.29 \\
    \bottomrule
    \end{tabular}
    \vspace{2mm}
    \caption{Indexing time (s). The Synthetic data has $n=2^{16}$, $d=32$}
    \label{tab:index-time}
\end{minipage}
\vspace{-3mm}
\end{table}
\vspace{-6mm}
\section{Advantages and Limitations}\label{sec:discussion}
\vspace{-2mm}
The HENN structure is \underline{easy to implement} and closely resembles HNSW, with a key difference in how layers are constructed. It ensures that each layer forms an $\eps$-net, which can be done by random sampling.
HENN is also \underline{modular} with respect to the choice of navigable graph used at each layer. Any graph structure suitable for similarity search, such as k-NN, Delaunay Triangulation (DT), or Navigable Small World (NSW), can be integrated into HENN.

HENN supports \underline{parallelization} during the indexing phase, enabling faster construction on large datasets. It also accommodates \underline{dynamic updates}, both insertions and deletions, while preserving theoretical guarantees. Further details are provided in Appendices~\ref{sec:app:par} and~\ref{sec:app:dynamic}.

The HENN structure is applicable to any metric space with induced bounded VC-dimension of $\Theta(d)$ (like all widely-used $\ell_p$-norm and angular distances for ANN). However, a limitation is that the theoretical guarantees no longer hold when the VC-dimension is \underline{unbounded}.

The provided time guarantees of HENN depend on the properties of \underline{the used navigable graphs}, such as the degree and recall bounds. In practice, commonly used graph structures often have constant degree and recall bounds. However, some graphs, such as Delaunay Triangulation (DT), can have degrees that grow exponentially with the dimension. As a result, this can cause a limitation for this structure.
These theoretical guarantees come at the cost of \underline{increased indexing time}. The construction time grows with the number of sampling iterations required to ensure that each layer forms an $\eps$-net.
\vspace{-3mm}
\section{Conclusion}
\vspace{-4mm}
We introduced HENN, a hierarchical graph-based structure for ANN search that combines theoretical guarantees with practical efficiency. By constructing $\eps$-net-based layers, HENN achieves provable polylogarithmic query time while remaining simple to implement. We also provide a probabilistic analysis for HNSW, offering insight into its empirical success. Experiments show that HENN performs robustly across both standard and adversarial datasets, making it a practical and scalable solution for nearest neighbor search.

\bibliography{ref}

\begin{thebibliography}{10}

\bibitem{aumuller2020ann}
Martin Aum{\"u}ller, Erik Bernhardsson, and Alexander Faithfull.
\newblock Ann-benchmarks: A benchmarking tool for approximate nearest neighbor algorithms.
\newblock {\em Information Systems}, 87:101374, 2020.

\bibitem{bentley1975multidimensional}
Jon~Louis Bentley.
\newblock Multidimensional binary search trees used for associative searching.
\newblock {\em Communications of the ACM}, 18(9):509--517, 1975.

\bibitem{beygelzimer2006cover}
Alina Beygelzimer, Sham Kakade, and John Langford.
\newblock Cover trees for nearest neighbor.
\newblock In {\em Proceedings of the 23rd international conference on Machine learning}, pages 97--104, 2006.

\bibitem{chazelle2000discrepancy}
Bernard Chazelle.
\newblock {\em The Discrepancy Method: Randomness and Complexity}.
\newblock Cambridge University Press, Cambridge, UK, 2000.

\bibitem{datar2004locality}
Mayur Datar, Nicole Immorlica, Piotr Indyk, and Vahab~S Mirrokni.
\newblock Locality-sensitive hashing scheme based on p-stable distributions.
\newblock In {\em Proceedings of the twentieth annual symposium on Computational geometry}, pages 253--262, 2004.

\bibitem{deberg2008computational}
Mark de~Berg, Otfried Cheong, Marc van Kreveld, and Mark Overmars.
\newblock {\em Computational Geometry: Algorithms and Applications}.
\newblock Springer-Verlag, Berlin, Heidelberg, 3rd edition, 2008.

\bibitem{fan2024survey}
Wenqi Fan, Yujuan Ding, Liangbo Ning, Shijie Wang, Hengyun Li, Dawei Yin, Tat-Seng Chua, and Qing Li.
\newblock A survey on rag meeting llms: Towards retrieval-augmented large language models.
\newblock In {\em Proceedings of the 30th ACM SIGKDD Conference on Knowledge Discovery and Data Mining}, pages 6491--6501, 2024.

\bibitem{ge2013optimized}
Tiezheng Ge, Kaiming He, Qifa Ke, and Jian Sun.
\newblock Optimized product quantization.
\newblock {\em IEEE transactions on pattern analysis and machine intelligence}, 36(4):744--755, 2013.

\bibitem{gibbons2002fast}
Phillip~B Gibbons, Yossi Matias, and Viswanath Poosala.
\newblock Fast incremental maintenance of approximate histograms.
\newblock {\em ACM Transactions on Database Systems (TODS)}, 27(3):261--298, 2002.

\bibitem{han2023comprehensive}
Yikun Han, Chunjiang Liu, and Pengfei Wang.
\newblock A comprehensive survey on vector database: Storage and retrieval technique, challenge.
\newblock {\em arXiv preprint arXiv:2310.11703}, 2023.

\bibitem{harpeled2011geometric}
Sariel Har-Peled.
\newblock {\em Geometric Approximation Algorithms}, volume 173 of {\em Mathematical Surveys and Monographs}.
\newblock American Mathematical Society, Providence, RI, 2011.

\bibitem{har2012approximate}
Sariel Har-Peled, Piotr Indyk, and Rajeev Motwani.
\newblock Approximate nearest neighbor: Towards removing the curse of dimensionality.
\newblock 2012.

\bibitem{haussler1986epsilon}
David Haussler and Emo Welzl.
\newblock Epsilon-nets and simplex range queries.
\newblock In {\em Proceedings of the second annual symposium on Computational geometry}, pages 61--71, 1986.

\bibitem{healy2024uniform}
John Healy and Leland McInnes.
\newblock Uniform manifold approximation and projection.
\newblock {\em Nature Reviews Methods Primers}, 4(1):82, 2024.

\bibitem{indyk1998approximate}
Piotr Indyk and Rajeev Motwani.
\newblock Approximate nearest neighbors: towards removing the curse of dimensionality.
\newblock In {\em Proceedings of the thirtieth annual ACM symposium on Theory of computing}, pages 604--613, 1998.

\bibitem{indyk2023worst}
Piotr Indyk and Haike Xu.
\newblock Worst-case performance of popular approximate nearest neighbor search implementations: Guarantees and limitations.
\newblock {\em Advances in Neural Information Processing Systems}, 36:66239--66256, 2023.

\bibitem{jayaram2019diskann}
Suhas Jayaram~Subramanya, Fnu Devvrit, Harsha~Vardhan Simhadri, Ravishankar Krishnawamy, and Rohan Kadekodi.
\newblock Diskann: Fast accurate billion-point nearest neighbor search on a single node.
\newblock {\em Advances in neural information processing Systems}, 32, 2019.

\bibitem{jegou2010product}
Herve Jegou, Matthijs Douze, and Cordelia Schmid.
\newblock Product quantization for nearest neighbor search.
\newblock {\em IEEE transactions on pattern analysis and machine intelligence}, 33(1):117--128, 2010.

\bibitem{jegou2011product}
Herv{\'e} J{\'e}gou, Matthijs Douze, and Cordelia Schmid.
\newblock Product quantization for nearest neighbor search.
\newblock In {\em IEEE Transactions on Pattern Analysis and Machine Intelligence (TPAMI)}, volume~33, pages 117--128. IEEE, 2011.

\bibitem{jing2024large}
Zhi Jing, Yongye Su, Yikun Han, Bo~Yuan, Haiyun Xu, Chunjiang Liu, Kehai Chen, and Min Zhang.
\newblock When large language models meet vector databases: A survey.
\newblock {\em arXiv preprint arXiv:2402.01763}, 2024.

\bibitem{johnson1984extensions}
William~B Johnson, Joram Lindenstrauss, et~al.
\newblock Extensions of lipschitz mappings into a hilbert space.
\newblock {\em Contemporary mathematics}, 26(189-206):1, 1984.

\bibitem{krauthgamer2004navigating}
Robert Krauthgamer and James~R Lee.
\newblock Navigating nets: Simple algorithms for proximity search.
\newblock In {\em Proceedings of the fifteenth annual ACM-SIAM symposium on Discrete algorithms}, pages 798--807. Citeseer, 2004.

\bibitem{le1993correlations}
G{\'e}rard Le~Caer and Renaud Delannay.
\newblock Correlations in topological models of 2d random cellular structures.
\newblock {\em Journal of Physics A: Mathematical and General}, 26(16):3931, 1993.

\bibitem{lecun1998gradient}
Yann LeCun, L{\'e}on Bottou, Yoshua Bengio, and Patrick Haffner.
\newblock Gradient-based learning applied to document recognition.
\newblock {\em Proceedings of the IEEE}, 86(11):2278--2324, 1998.

\bibitem{lewis2020retrieval}
Patrick Lewis, Ethan Perez, Aleksandra Piktus, Fabio Petroni, Vladimir Karpukhin, Naman Goyal, Heinrich K{\"u}ttler, Mike Lewis, Wen-tau Yih, Tim Rockt{\"a}schel, et~al.
\newblock Retrieval-augmented generation for knowledge-intensive nlp tasks.
\newblock {\em Advances in neural information processing systems}, 33:9459--9474, 2020.

\bibitem{li2019approximate}
Wen Li, Ying Zhang, Yifang Sun, Wei Wang, Mingjie Li, Wenjie Zhang, and Xuemin Lin.
\newblock Approximate nearest neighbor search on high dimensional data—experiments, analyses, and improvement.
\newblock {\em IEEE Transactions on Knowledge and Data Engineering}, 32(8):1475--1488, 2019.

\bibitem{liu2021revisiting}
Yingfan Liu, Cheng Hong, and Jiangtao Jiang.
\newblock Revisiting k-nearest neighbor graph construction on high-dimensional data: Experiments and analyses.
\newblock {\em arXiv preprint arXiv:2112.02234}, 2021.

\bibitem{lu2021hvs}
Kejing Lu, Mineichi Kudo, Chuan Xiao, and Yoshiharu Ishikawa.
\newblock Hvs: hierarchical graph structure based on voronoi diagrams for solving approximate nearest neighbor search.
\newblock {\em Proceedings of the VLDB Endowment}, 15(2):246--258, 2021.

\bibitem{malkov2018efficient}
Yu~A Malkov and Dmitry~A Yashunin.
\newblock Efficient and robust approximate nearest neighbor search using hierarchical navigable small world graphs.
\newblock {\em IEEE transactions on pattern analysis and machine intelligence}, 42(4):824--836, 2018.

\bibitem{malkov2014approximate}
Yury Malkov, Alexander Ponomarenko, Andrey Logvinov, and Vladimir Krylov.
\newblock Approximate nearest neighbor algorithm based on navigable small world graphs.
\newblock {\em Information Systems}, 45:61--68, 2014.

\bibitem{munoz2019hierarchical}
Javier~Vargas Munoz, Marcos~A Gon{\c{c}}alves, Zanoni Dias, and Ricardo da~S Torres.
\newblock Hierarchical clustering-based graphs for large scale approximate nearest neighbor search.
\newblock {\em Pattern Recognition}, 96:106970, 2019.

\bibitem{omohundro1989five}
Stephen~M Omohundro.
\newblock Five balltree construction algorithms.
\newblock 1989.

\bibitem{ozan2016competitive}
Ezgi~Can Ozan, Serkan Kiranyaz, and Moncef Gabbouj.
\newblock Competitive quantization for approximate nearest neighbor search.
\newblock {\em IEEE Transactions on Knowledge and Data Engineering}, 28(11):2884--2894, 2016.

\bibitem{pan2024survey}
James~Jie Pan, Jianguo Wang, and Guoliang Li.
\newblock Survey of vector database management systems.
\newblock {\em The VLDB Journal}, 33(5):1591--1615, 2024.

\bibitem{pennington2014glove}
Jeffrey Pennington, Richard Socher, and Christopher~D Manning.
\newblock Glove: Global vectors for word representation.
\newblock In {\em Proceedings of the 2014 Conference on Empirical Methods in Natural Language Processing (EMNLP)}, pages 1532--1543. Association for Computational Linguistics, 2014.

\bibitem{pham2022falconn++}
Ninh Pham and Tao Liu.
\newblock Falconn++: A locality-sensitive filtering approach for approximate nearest neighbor search.
\newblock {\em Advances in Neural Information Processing Systems}, 35:31186--31198, 2022.

\bibitem{prokhorenkova2020graph}
Liudmila Prokhorenkova and Aleksandr Shekhovtsov.
\newblock Graph-based nearest neighbor search: From practice to theory.
\newblock In {\em International Conference on Machine Learning}, pages 7803--7813. PMLR, 2020.

\bibitem{van2008visualizing}
Laurens Van~der Maaten and Geoffrey Hinton.
\newblock Visualizing data using t-sne.
\newblock {\em Journal of machine learning research}, 9(11), 2008.

\bibitem{vitter1985random}
Jeffrey~S Vitter.
\newblock Random sampling with a reservoir.
\newblock {\em ACM Transactions on Mathematical Software (TOMS)}, 11(1):37--57, 1985.

\bibitem{wang2021comprehensive}
Mengzhao Wang, Xiaoliang Xu, Qiang Yue, and Yuxiang Wang.
\newblock A comprehensive survey and experimental comparison of graph-based approximate nearest neighbor search.
\newblock {\em arXiv preprint arXiv:2101.12631}, 2021.

\end{thebibliography}
\newpage
\appendix
\section*{Appendix}
\subsection*{Table of Content}

\noindent
\noindent
\makebox[\linewidth][l]{1.\quad Related Work \dotfill\ \ref{sec:related}}\\
\makebox[\linewidth][l]{2.\quad $\eps$-Net and Navigable Graphs \dotfill\ \ref{sec:buildcomponents}}\\
\makebox[\linewidth][l]{3.\quad Dynamic Setting \dotfill\ \ref{sec:app:dynamic}}\\
\makebox[\linewidth][l]{4.\quad Parallelization \dotfill\ \ref{sec:app:par}}\\
\makebox[\linewidth][l]{5.\quad Proofs \dotfill\ \ref{sec:app:proofs}}\\
\makebox[\linewidth][l]{6.\quad Pseudo-codes \dotfill\ \ref{sec:app:code}}\\
\makebox[\linewidth][l]{7.\quad More on Experiments \dotfill\ \ref{sec:app:exp}}\\
\section{Related Work}\label{sec:related}
In this section, we review the literature most relevant to our work. A broader overview of techniques for the Approximate Nearest Neighbor (ANN) problem, including references to comprehensive surveys, is presented in the Introduction section. Here, we focus primarily on \emph{hierarchical} approaches developed for solving the ANN problem.

A notable class of approaches for solving the ANN problem is based on hierarchical structures. One of the most widely used methods in this category is Hierarchical Navigable Small World (\textbf{HNSW})~\cite{malkov2018efficient}, which constructs a multi-layered structure of navigable small-world graphs to enable efficient search. Interestingly, the core idea of hierarchical organization can be traced back to earlier work in computational geometry, including \textbf{Cover Trees}~\cite{beygelzimer2006cover} and \textbf{Navigating Nets}~\cite{krauthgamer2004navigating}.

\paragraph{HNSW.}
Hierarchical Navigable Small World (HNSW) graphs~\cite{malkov2018efficient} construct a multi-layered hierarchy of navigable small world (NSW) graphs. An NSW graph serves as an efficient approximation of the Delaunay graph~\cite{harpeled2011geometric}, which is known to be an optimal structure for solving the approximate nearest neighbor (ANN) problem~\cite{malkov2014approximate}. The Delaunay graph is closely related to the Voronoi diagram, which partitions the space into cells based on their proximity to the points in the dataset. Unlike the Delaunay graph, which requires explicit geometric computation, NSW graphs are built incrementally by inserting points and connecting them to their approximate nearest neighbors~\cite{malkov2014approximate}. During insertion and querying, HNSW employs a greedy search algorithm (Algorithm~\ref{alg:query}) to navigate through the graph and locate nearby points. To improve scalability, HNSW organizes the data in a hierarchy where the number of nodes decreases exponentially across layers, resulting in $O(\log n)$ layers and a total space complexity of $O(n)$.

\paragraph{Cover Trees and Navigating Nets.} The Cover Tree data structure~\cite{beygelzimer2006cover} provides a scalable and theoretically grounded approach to nearest neighbor (NN) search in general metric spaces. The cover tree maintains logarithmic query time and linear space complexity under a measure of intrinsic data dimensionality. The structure recursively organizes points into a nested hierarchy of layers, enforcing both covering and separation invariants, which enable efficient traversal and pruning during NN queries, inspired by Navigating Nets. Navigating Nets~\cite{krauthgamer2004navigating} is another hierarchical structure based on $r$-nets designed for efficient proximity search in general metric spaces. The approach relies on constructing a sequence of nested $r$-nets, which progressively approximate the dataset at multiple scales. $r$-net is a subset of points that cover an $r$ distance around any points in the dataset. These hierarchical nets enable fast navigation by guiding the search from coarse to fine resolutions. It achieves logarithmic query and update times with respect to the number of points. Unlike randomized methods, Navigating Nets gives predictable performance. However, they are difficult to implement and do not scale well to large datasets. Because of this, they are rarely used in practical nearest neighbor applications today.

\paragraph{General Solutions to Nearest Neighbor Search.}
Solutions to the nearest neighbor problem can be categorized along several dimensions. One common distinction is between classical methods, which primarily target the exact NN problem. Examples include k-d trees~\cite{bentley1975multidimensional}, ball trees~\cite{omohundro1989five}, and Delaunay triangulations~\cite{har2012approximate}. While effective in low-dimensional spaces, these methods typically fail to scale in high-dimensional settings. Another broad category includes quantization-based methods~\cite{jegou2010product, ge2013optimized, ozan2016competitive}, which cluster the data and represent points by their assigned centroids (codewords), thereby approximating distances efficiently. Additionally, hashing-based methods such as Locality-Sensitive Hashing (LSH~\cite{datar2004locality}) provide theoretical guarantees and have been widely used for high-dimensional ANN, and finally, the graph-based methods~\cite{malkov2014approximate, malkov2018efficient, prokhorenkova2020graph, lu2021hvs}. For comprehensive overviews of these and other approaches, we refer the reader to the following surveys~\cite{pan2024survey, han2023comprehensive, aumuller2020ann, wang2021comprehensive, li2019approximate, liu2021revisiting}.
\section{\enet and Navigable Graphs}\label{sec:buildcomponents}

In this section, we provide background on the theory of $\eps$-nets, covering both sampling-based and discrepancy-based algorithms for constructing them. We also describe the heuristics used in our experiments for building $\eps$-nets. Finally, we discuss the background on navigation graphs that can be used in the HENN data structure.

\subsection{$\eps$-nets}\label{sec:app:build-epsnet}
First, we define \enet for a general range system $(\points, \mathcal{R})$:

\begin{definition}[\enet]
    Given a range space defined as $(\points, \mathcal{R})$, where $\points$ is a set of points and $\mathcal{R}$ is a family of subsets (ranges) over $\points$, an $\eps$-net for this system is a subset $\mathcal{A} \subseteq \points$ such that:
    \[
        \forall R \in \mathcal{R}: \quad \frac{|\points \cap R|}{|\points|} > \varepsilon \ \Rightarrow\ R \cap \mathcal{A} \neq \varnothing.
    \]
    In other words, the set $\mathcal{A}$ intersects all \emph{heavy ranges}, those that contain more than an $\eps$ fraction of the total points in $\points$.
\end{definition}

The popular \enet theorem is as follows~\cite{har2012approximate, haussler1986epsilon}:

\begin{theorem}[\enet Theorem]\label{thm:enet}
    For a range space defined like above, with VC-dimension $\delta$, a random independent sample of size $s$ will give an \enet with probability more than $1 - \varphi$, where $s$ is:
    \begin{align}\label{eq:enet:sizes}
        s = \max \left(\frac{1}{\eps} \log \frac{1}{\varphi} + \frac{\delta}{\eps}\log \frac{\delta}{\eps}\right)
    \end{align}
\end{theorem}

There are two popular algorithms for building an $\eps$-net of size $O(\frac{\delta}{\eps}\log \frac{\delta}{\eps})$ for a range space $(\points, \mathcal{R})$ of VC-dimension $\delta$. A randomized {\bf sampling-based} and a deterministic {\bf discrepancy-based} algorithm.

\paragraph{Sampling-based Construction.} 
A randomized Monte-Carlo algorithm constructs an \enet by randomly drawing a set of samples from $X$, based on the value $s$ in Theorem~\ref{thm:enet}. 

\begin{corollary}
    For a constant failure probability $\varphi \leq \frac{1}{2}$, the sampling-based algorithm constructs an $\eps$-net of size $O(\frac{\delta}{\eps}\log \frac{\delta}{\eps})$ in time $O(\frac{\delta}{\eps}\log \frac{\delta}{\eps})$.
\end{corollary}

Indeed, a single run of the Monte-Carlo algorithm can fail to generate an \enet. 
Therefore, we use it within a Las-Vegas randomized algorithm (Algorithm~\ref{alg:epsnet-rand}) that repeats drawing samples until an \enet is found. 

\begin{algorithm}[ht]
\caption{Building $\eps$-net (Sampling-based Algorithm)}
\label{alg:epsnet-rand}
\begin{algorithmic}[1]
\Require The range space $(\points, \mathcal{R})$, value of $\eps$, failure probability $\varphi$, and the exponential decay $m$.
\Ensure The $\eps$-net $\mathcal{A}$.
\Function{BuildEpsNetSampling}{$\points, \eps, \varphi$}\Comment{Sampling-based algorithm}
    \Repeat
    \State $s \gets $ calculate the size (Equation~\ref{eq:enet:sizes})
    \State $\mathcal{A} \gets$ $s$ random samples with replacement from $\points$.
    \Until{{\sc IsEpsNet}$(\mathcal{A})= $ {\bf true}}
    \State {\bf Return} $\mathcal{A}$
\EndFunction
\end{algorithmic}
\end{algorithm}

\paragraph{Practical Implementation.}\label{par:practice:enet}
Algorithm~\ref{alg:epsnet-rand} relies on a function {\sc IsEpsNet} that, given a set $\mathcal{A}$, verifies if it is an \enet.
Such a verification would require enumerating the collection of valid ranges, i.e., in our case, all ring ranges that contain at least $\eps$ elements ($O(n^\delta)$).
Therefore, instead of verification over the entire collection, we sample a subset of ranges, uniformly at random, and verify whether $\mathcal{A}$ contains at least one element from each of the sampled ranges. The number of sample ranges is treated as a hyperparameter in our construction algorithm. 

Let $T_v$ be the time required to verify if a set $\mathcal{A}\subset X$ is an \enet.
The expected number of time the Las-Vegas algorithm requires to repeat the sampling is $\frac{1}{1-\varphi}$.
Hence, fixing the failure probability to $\varphi = \frac{1}{2}$, the expected time complexity of the algorithm is \(O\left( \frac{\delta}{\eps}\log \frac{\delta}{\eps} + T_v \right)\).

\begin{algorithm}[ht]
\caption{Building $\eps$-net (Discrepancy-based Algorithm)}
\label{alg:epsnet}
\begin{algorithmic}[1]
\Require The range space $(\points, \mathcal{R})$, value of $\eps$, failure probability $\varphi$, and the exponential decay $m$.
\Ensure The $\eps$-net $\mathcal{A}$.
\Function{BuildEpsNetDiscV1}{$\points, \eps, \mathcal{R}$}\Comment{Build $\eps$-net by providing $\eps$}
    \State $k \gets$ number of iterations (Equation~\ref{eq:disc-iters})
    \State $\mathcal{A} \gets \points$
    \For{$1 \leq i \leq k$}
        \State $\mathcal{A} \gets \textit{Halving}(\mathcal{A}, \mathcal{R})$\Comment{The halving step, with arbitrary matching.}
    \EndFor
    \State {\bf Return} $\mathcal{A}$
\EndFunction
\Function{BuildEpsNetDiscV2}{$\points, m, \mathcal{R}$}\Comment{Build $\eps$-net by providing $m$}
    \State $k \gets m$\Comment{Only continue $m$ times to reduce the size by $2^m$}
    \State $\mathcal{A} \gets \points$
    \For{$1 \leq i \leq k$}
        \State $\mathcal{A} \gets \textit{Halving}(\mathcal{A}, \mathcal{R})$\Comment{The halving step, with arbitrary matching.}
    \EndFor
    \State {\bf Return} $\mathcal{A}$
\EndFunction
\end{algorithmic}
\end{algorithm}

\paragraph{Discrepancy-based Construction.} 
In the following, we provide a high-level overview of the Discrepancy-based algorithm, and refer the reader to~\cite{harpeled2011geometric, chazelle2000discrepancy} for more details.
The algorithm (Algorithm~\ref{alg:epsnet}) works by iteratively {\it halving} the point-set $\points$, until reaching the desired size of $c_0 \frac{\delta}{\eps} \log \frac{\delta}{\eps}$, where $c_0$ is a large enough constant. 

In order to do so, it first constructs an arbitrary matching $\Pi$ of the points. A matching $\Pi$ is a set of pairs $(x, y)$ where $x, y\in\points$, it also partitions $\points$ into a set of $\frac{|\points|}{2}$ disjoint pairs. Given this matching, this algorithm randomly picks one of the points in each pair, removing the other point of the pair, resulting in a subset of remaining points $\points_1 \subset \points$ where $|\points_1| = \frac{|\points|}{2}$.

Continuing this process $k$ times for the following value of $k$

\begin{align}\label{eq:disc-iters}
    2^k = \frac{|\points|}{c_0 \frac{\delta}{\eps} \log \frac{\delta}{\eps}} 
\end{align}

results in the set $|\points_k|$ which is an $\eps$-net for $\points$. It is easy to make this process deterministic by following the conditional expectation method at each halving step~\cite{chazelle2000discrepancy}. By following a Sketch-and-Merge algorithm, one can achieve a near-linear running time for this construction~\cite{har2012approximate}:

\begin{corollary}
    The discrepancy-based algorithm finds an $\eps$-net of size $O(\frac{\delta}{\eps}\log \frac{\delta}{\eps})$ in time $O(\delta^{3\delta} \cdot (\frac{1}{\eps}\log \frac{\delta}{\eps})^\delta \cdot n)$.
\end{corollary}

\paragraph{Comparison.} 
The randomized algorithm is straightforward to implement, requiring only random sampling from each layer $\mathcal{L}_i$ to construct the subsequent layer $\mathcal{L}_{i+1}$. However, it is inherently randomized and provides running-time guarantees in expectation. Furthermore, its time complexity depends on the time to verify if the selected set is indeed an \enet. 

In contrast, the deterministic discrepancy-based algorithm deterministically constructs an $\eps$-net by progressively halving each layer $\mathcal{L}_i$. With the hyperparameter $m$, as the exponential decay of the HENN graph, this process involves only halving $\mathcal{L}_i$ up to $m$ times to identify the next layer $\mathcal{L}_{i+1}$ (See BuildEpsNetDiscV2 in Algorithm~\ref{alg:epsnet}). Nevertheless, the running time of the discrepancy-based algorithm depends on the dimensionality of the input points. A pseudo-code of these two algorithms is provided in Algorithm~\ref{alg:epsnet}.

\subsection{Navigable Graphs}\label{sec:app:nav-graph}
Graph-based algorithms for the ANN problem typically begin by constructing a graph on the given dataset $\points$. A key property of these graphs is \emph{navigability}, which ensures that the Greedy Search algorithm (Algorithm~\ref{alg:query}) can be effectively applied~\cite{malkov2014approximate}. Specifically, navigability means that by following a sequence of locally greedy steps, the algorithm can successfully reach an approximate nearest neighbor of the query point $q$.

According to this definition, a \emph{complete graph} over the point set is trivially navigable. The most optimized navigable graph can, in principle, be obtained by constructing the dual of the Voronoi diagram of the points, known as the \emph{Delaunay triangulation}~\cite{deberg2008computational}. However, constructing this graph is computationally challenging, particularly in high dimensions, due to the curse of dimensionality.

An efficient approximation of the Delaunay triangulation can be achieved through a simple randomized algorithm that incrementally inserts points and connects each new point to its nearest neighbors in the existing graph structure. This approach forms the basis of the \emph{Navigable Small World} (NSW) graph~\cite{malkov2014approximate}. The HNSW algorithm adopts a similar strategy within each layer, while introducing additional heuristics to improve practical performance, such as adding random exploration edges between points. These heuristics can also be incorporated into the HENN structure as well, treating the navigable graph as a black-box~\cite{malkov2018efficient}.

A $k$-NN graph can also be used; however, it is well known that for small values of $k$, such graphs tend to exhibit numerous local minima~\cite{liu2021revisiting}.

As previously discussed, any black-box navigable graph constructed over the points within a single layer can be easily integrated into the HENN framework. We refer the reader to Appendix~\ref{sec:app:exp} for a comparative analysis of different navigation graph choices integrated within HENN.

\subsubsection{Navigation Graph Construction via Dimensionality Reduction}
We introduce a heuristic that leverages dimensionality reduction to construct a navigation graph for each individual layer of HENN. The procedure begins by reducing the dimensionality of the data points in a given layer to a low-dimensional space (typically 2D or 3D). Subsequently, we compute the exact Delaunay triangulation (DT) over the reduced representation and construct this graph on the point set. This approach is computationally efficient, as DT construction in low dimensions is fast and does not suffer from the challenges associated with high-dimensional geometric computations.

For dimensionality reduction, we employ techniques that aim to preserve pairwise distances, such as t-SNE~\cite{van2008visualizing}. Additional techniques and details regarding this heuristic are discussed in Appendix~\ref{sec:app:exp}.

\section{Dynamic Setting}\label{sec:app:dynamic}
In this section, we present a procedure for maintaining the HENN structure under dynamic updates to the point set $\points$. The supported operations include \texttt{Insert(x)}, which adds a new point $x$, and \texttt{Delete(x)}, which removes an existing point $x \in \points$.

Based on the $\eps$-net construction described in Section~\ref{sec:app:build-epsnet}, a random sample of an appropriate size forms an $\eps$-net of $\points$ with high probability. We denote this required sample size by $f_\eps$, given by:

\[
    f_\eps = O\left(\frac{d}{\eps} \log \frac{d}{\eps}\right)
\]

Each layer of HENN, denoted by $\mathcal{L}_i$, is constructed as a random sample of size $f_{\eps(|\mathcal{L}_{i-1}|)}$ from the previous layer (see Equation~\ref{eq:eps}). Consequently, the problem reduces to dynamically maintaining a random sample $S$ of size $f$ from the point set $\points$ \footnote{$S$ is a random sample with replacement, with each element having a probability $\frac{f_\eps}{n}$ being in $S$.}.

This problem can be addressed using Reservoir Sampling~\cite{vitter1985random} and the \emph{Backing Samples} technique~\cite{gibbons2002fast}. The key idea is to handle \texttt{Insert(x)} operations by probabilistically adding the new element to the sample $S$ using a non-uniform coin toss. To support deletions, a larger backing sample is maintained beyond size $f$, allowing for efficient resampling of $S$ once the size drops below a threshold. This approach yields a constant amortized update time.

According to this, we can maintain the HENN structure dynamically:

\paragraph{\texttt{Insert(x)}:} To insert a new point, dynamic updates are performed starting from layer $\mathcal{L}_1$ and proceeding upward through the hierarchy, stopping at the highest layer where the new point is included. This process takes \( O(\log n) \) time, matching the insertion time complexity of HNSW.

\paragraph{\texttt{Delete(x)}:} Deletion begins at layer $\mathcal{L}_1$, where the point $x$ is removed if present. If the size of a layer falls below a critical threshold (as discussed in~\cite{gibbons2002fast}), the layer must be resampled. Following resampling, the HENN structure is rebuilt from that layer up to the root, which incurs a cost of \( O(n \log n) \) in the worst case. However, since such rebuilding occurs infrequently, only when the layer size drops significantly (e.g., \( f < c_0 n \) for a constant $c_0$), the amortized cost remains \( O(\log n) \).

\section{Parallelization}\label{sec:app:par}
In this section, we present a parallelized approach to constructing the HENN index during preprocessing. While the original HENN construction, shown in Algorithm~\ref{alg:preprocess}, runs sequentially by building layers $\mathcal{L}_1$ through $\mathcal{L}_m$ from the base point set $\points$, this process can be parallelized to reduce preprocessing time.

To enable parallelization, we exploit the fact that layer sampling in HENN is performed with replacement. Given $p$ parallel CPU cores, we can independently generate samples for each layer in parallel. Specifically, each core performs independent sampling, effectively achieving a $p$-factor speedup for the sampling phase performed on each layer.

HNSW also uses a parallelization to enhance preprocessing~\cite{malkov2018efficient}, where each input point is processed independently. For each point, a random level is assigned, and the point is inserted into the corresponding layers. This results in a total construction time of $O(n \log n)$, which can be reduced to $O(\frac{n \log n}{p})$ under parallel execution with $p$ cores.

For HENN, the construction begins by sampling a subset of size $\frac{n}{2^m}$ for the first layer $\mathcal{L}_1$, and recursively building higher layers. Each layer $\mathcal{L}_i$ requires constructing a navigation graph, the complexity of which depends on the chosen method. For instance, NSW-based graph construction requires $O(|\mathcal{L}_i| \log |\mathcal{L}_i|)$ time on the $i$th layer. Excluding graph construction, the sampling phase alone can be executed in $O(\frac{n \log n}{p})$ time using $p$ cores.
\section{Proofs}\label{sec:app:proofs}
\subsection{Proof of Lemma~\ref{lem:layer-size}}
\begin{proof}
    Fix a layer $i$, let $s = |\mathcal{L}_{i-1}|$. The size of an $\eps$-net according to the described procedure is $O(\frac{d}{\eps}\log \frac{d}{\eps})$. As a result, for the layer $\mathcal{L}_i$ and a large enough constant $c_1 < c_0$, we have:

    \begin{align*}
        |\mathcal{L}_i| &\leq c_1 \frac{d}{\eps(s)} \log\frac{d}{\eps(s)} \\
        &=c_1 \frac{s}{c_0 2^{m} \log s} \log \frac{s}{c_0 2^{m} \log s}\\
        &\leq c_1 \frac{s}{c_0 2^{m} \log s} \log \frac{s}{2^{m}}\\
        &= c_1 \frac{s}{c_0 2^{m}} \leq \frac{s}{2^{m}}
    \end{align*}
\end{proof}

\subsection{Proof of Lemma~\ref{lem:two-layer-case}}
\begin{proof}
    Since $\mathcal{L}$ is an $\eps$-net of $\points$, we know that each (ring) range $R$ of size more than $\eps \cdot n$, intersects with $\mathcal{L}$. In other words:

    \[
        |R| \geq \eps n \xrightarrow{} \mathcal{L} \cap R \neq \varnothing
    \]

    This comes from the definition of an $\eps$-net. 
    
    Based on the definition of Recall Bound, we know that there are at most $\rho_\delta$ more points, denoted as $P_{\mathcal{L} }= \{p_1, p_2, \cdots, p_{\rho_\delta}\}$, in $\mathcal{L}$ that are closer to $q$ than $\overline{p}$ (with probability more than $\delta$). 
    
    Assume that the points in $P_{\mathcal{L}}$ are sorted based on their distance to $q$, with $p_1$ being the closest. 
    For each $p_j \in P_{\mathcal{L}}$ define a unique range $R_j$, which is a ring centered at $q$, covering all the distances between $\left(dist(q, p_{j-1}), dist(q, p_j)\right)$, exclusively (see Figure~\ref{fig:lem:two-layer-case}). In addition, define one more ring, $R_{\rho_\delta + 1}$ for $\left(dist(q, p_{\rho_\delta}), dist(q, \overline{p})\right)$.

    All these $\rho_\delta +1$ ranges are disjoint, and they do not contain any point in $\mathcal{L}$. Since $\mathcal{L}$ is an $\eps$-net for this range space, this means that all the ranges $R_j$ have at most $\eps \cdot n$ points from $\points$ (the lower level). As a result, the union of all these ranges contains at most $\eps \cdot n \cdot (\rho_\delta + 1)$ points.
    
    By considering the $P_\mathcal{L}$ itself, this proves the fact that there are at most $\eps \cdot n \cdot (\rho_\delta + 2)$ points in $\points$ that are closer to $q$ than $\overline{p}$.
\end{proof}

\begin{figure}[h]
    \centering
    \includegraphics[width=0.3\linewidth]{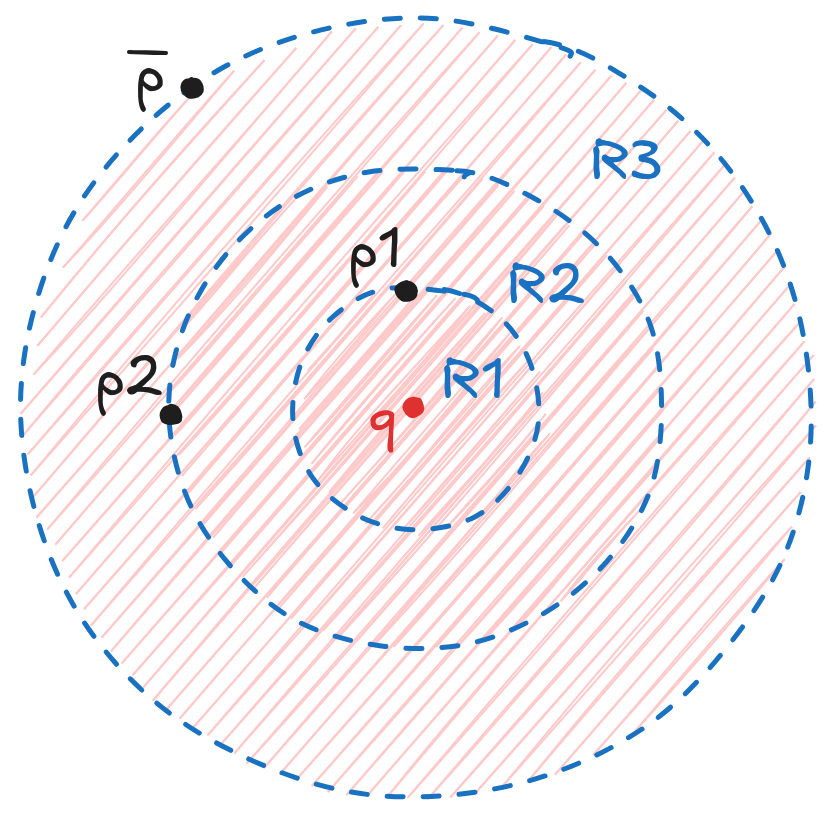}
    \caption{Visualization of Lemma~\ref{lem:two-layer-case}. In this example, $\rho_\delta = 2$ and the black points are inside the $\eps$-net.}
    \label{fig:lem:two-layer-case}
\end{figure}

\subsection{Proof of Theorem~\ref{thm:time}}
\begin{proof}
    We know that the bottom-most layer $\mathcal{L}_0$ is the point set itself, and the upper layer $\mathcal{L}_1$ is an $\eps$-net of $\mathcal{L}_0$ for $\eps = c_0 d\frac{\log n}{n} 2^m$ (see Equation~\ref{eq:eps}).

    The Query algorithm starts from the root and moves down to layer $\mathcal{L}_1$. Then, it finds a point $\overline{p}$ in $\mathcal{L}_1$ using the Greedy algorithm. 
    Finally, it applies the greedy algorithm in layer $\mathcal{L}_0$, starting from the point $\overline{p}$ and returns the result.

    Let $T(n)$ denote the running time of this algorithm for a point set of size $n$. The first step, starting from the root down to layer $\mathcal{L}_1$, is recursively equivalent to a search on a HENN graph, built on top of \underline{the points in $\mathcal{L}_1$}, with the same hyperparameters.
    As a result, using induction, the first step takes 
    \[T(|\mathcal{L}_1|) = T(\frac{n}{2^m})\]
    
    After finding $\overline{p}$ in layer $\mathcal{L}_1$, based on Lemma~\ref{lem:two-layer-case}, there are at most $\eps \cdot n \cdot (\rho_\delta + 2)$ points in $\points$ closer than $\overline{p}$ to $q$. We know that based on $\mathsf{GreedySearch}$, the second step, at each iteration, gets at least one step closer to $q$. Therefore, the total number of visited hops in the second step is at most $\eps \cdot n \cdot (\rho_\delta + 2)$. 
    
    Let $d^*$ denote the degree of each node in the navigation graph $\mathcal{G}_\points$. This results in the second step taking $O(\eps \cdot n \cdot \rho_\delta \cdot d^*)$ time to complete. Now, we can write the following recursion on the running time $T(n)$:

    \[
        T(n) \leq T(\frac{n}{2^m}) + \eps \cdot n \cdot (\rho_\delta + 2) \cdot d^* = T(\frac{n}{2^m}) + O(d \cdot \log n \cdot 2^m \cdot \rho_\delta \cdot d^*)
    \]

    The value $2^m$ is always kept as a constant (less than $64$ in HNSW). Hence, using the Master Theorem, the running time would be:
    \[
        T(n) = O(d \cdot d^* \cdot \rho_\delta \cdot \log^2 n)
    \]
\end{proof}
\section{Pseudo-codes}\label{sec:app:code}
\begin{algorithm}[h]
\caption{Greedy Search Algorithm}
\label{alg:query}
\begin{algorithmic}[1]
\Require The HENN graph $\mathcal{H}$ and the query point $q$.
\Ensure The approximate nearest neighbor of $q$ in $\points$.
\Function{Query}{$\mathcal{H}, q$}
    \State $v \gets \textit{random point from root} (\mathcal{L}_0)$
    \For{each layer $i = L, L-1, \cdots, 0$}
        \State $v \gets \textit{GreedySearch}(\mathcal{L}_i, v, q)$
    \EndFor
    \State {\bf Return} $v$
\EndFunction

\Function{GreedySearch}{$\mathcal{G}, v, q$}\Comment{$\mathcal{G}$ is the navigable graph, starting node $v$, and query $q$}
    \State $next \gets v$
    \Repeat
        \State $curr \gets next$
        \State $\mathcal{N} \gets $ neighbors of $curr$ in $\mathcal{G}$.
        \State $next \gets \arg\min_{u\in \mathcal{N}} dist(u, q)$\Comment{Closest neighbor to $q$}
    \Until{$dist(next, q) \geq dist(curr, q)$}\Comment{Until getting stuck in local minima}
    \State {\bf Return} $curr$
\EndFunction
\end{algorithmic}
\end{algorithm}
\section{More on Experiments}\label{sec:app:exp}
\subsection{Experimental Setup Configuration}\label{sec:app:exp:conf}
All experiments were conducted on an Ubuntu 24 server equipped with 64 CPU cores and 128 GB of RAM. Wherever applicable, we enabled hnswlib's parallel indexing capabilities for both the baseline HNSW and our proposed HENN implementations. To ensure consistent and interpretable timing results, all query operations were executed sequentially. The experiments on real-world benchmark datasets are mostly performed in Python\footnote{We used the Python binding of the C++ implementations, similar to hnswlib.}, while all other experiments were conducted in C++.

\subsection{Implementation Details}\label{sec:app:exp:impl}
The parameter $m$, which governs the exponential decay in layer size, is set to match the baseline configurations to ensure a fair comparison with an equal number of layers. To compute an $\eps$-net, we draw random samples with replacement and verify whether the sampled subset satisfies the $\eps$-net property. Otherwise, we repeat. 

A naive brute-force verification would require checking coverage over all possible ranges (e.g., all rings in the specified range space), which is computationally expensive. To make this process tractable in practice, we introduce a hyperparameter $r$, during the indexing phase, that denotes the number of ranges to check, also referred to as the \emph{range size}. During sampling, we randomly select $r$ ranges from the space and ensure that the sampled points hit all of them. Additional implementation details for this phase are provided in the supplemental material (also see background on building \enet in Section~\ref{sec:app:build-epsnet} in the Appendix).

\subsection{Methods and Datasets}\label{sec:app:exp:method}
In this section, we compare the HENN structure against the HNSW baseline. As discussed earlier, the primary distinction between the two lies in how hierarchical layers are constructed. HNSW selects layer points randomly, while HENN ensures that each layer forms an $\eps$-net of the previous one. In both methods, a navigable graph is built over each layer using an incremental greedy algorithm: when a new point is added, it is connected to the approximate nearest neighbors among the existing points in that layer, forming an NSW (navigable small world) graph. Other navigable graph constructions and baselines are discussed in the supplemental material.

To evaluate performance, we use both synthetic and real-world datasets. The synthetic data is generated by varying the number of dimensions and data points, with each coordinate sampled independently from an exponential distribution parameterized by $\lambda$. Larger $\lambda$ values produce more skewed distributions, making the dataset increasingly challenging and suitable for testing worst-case behavior. For real-world benchmarks, we use the SIFT~\cite{jegou2011product}, GloVe~\cite{pennington2014glove}, and MNIST~\cite{lecun1998gradient} datasets with Euclidean $\ell_2$-norm. The SIFT dataset consists of 1 million 128-dimensional visual descriptors extracted from image features, accompanied by 10,000 query vectors and their ground-truth nearest neighbors. The GloVe dataset contains 50,000 pre-trained 100-dimensional word embeddings derived from global word co-occurrence statistics across large text corpora. MNIST also contains 70,000 images (images) flattened into 784-dimensional vectors.. These are popular ANN benchmarks used in the literature\footnote{\href{https://ann-benchmarks.com}{ann-benchmarks.com}}. Additional experiments on more synthetic and real datasets are included in the supplemental material.

\subsection{More Experiments}
In this section, we present experiments on HENN integrated with various navigation graphs, as described in Appendix~\ref{sec:app:nav-graph}. Additionally, we explore the applicability of HENN under a commonly used distance metric in approximate nearest neighbor (ANN) problems: cosine similarity (also referred to as the angular distance metric). We evaluate the performance of HENN on two real-world datasets: FashionMNIST and GloVe, with the angular metric.

\paragraph{Navigation Graphs.}
As detailed in Appendix~\ref{sec:app:nav-graph}, different navigation graph constructions can be integrated into each individual layer of the HENN structure. In this section, we empirically evaluate the performance of the following navigation graph variants:

\begin{itemize}
    \item \textbf{K-Nearest Neighbor Graph ($\mathsf{knn}$):} Each point is connected to its exact $k$ nearest neighbors based on the original distance metric.
    
    \item \textbf{Navigable Small World Graph ($\mathsf{NSW}$):} This graph structure, used in HNSW, provides a navigable approximation of both the KNN and Delaunay Triangulation (DT) graphs.
    
    \item \textbf{Dimensionality Reduction-Based Graph ($\mathsf{dim\_red}$):} A dimensionality reduction technique, such as t-SNE~\cite{van2008visualizing}, PCA, the Johnson-Lindenstrauss (JL) projection~\cite{johnson1984extensions}, or UMAP~\cite{healy2024uniform}, is first applied to the data points, reducing them to a low-dimensional space (typically 2D or 3D). The exact DT graph is then constructed on the reduced representation.
\end{itemize}

\paragraph{Recall Bounds ($\rho_\delta$).} Figure~\ref{fig:recall-bound} presents a comparison of recall bounds across different navigation graph constructions. This metric serves as a key component in the time complexity guarantees established in Theorem~\ref{thm:time}. We conduct experiments across varying dataset sizes ($n$) and distributions. As expected, recall bounds increase with more skewed data distributions, e.g., higher values of $\lambda$ in exponential distributions, indicating a reduced probability of hitting a fixed number of nearest neighbors ($k$).

Notably, for both the KNN and NSW navigation graphs, the recall bound remains below 15 and exhibits minimal sensitivity to changes in dataset size. Among the dimensionality reduction-based graphs, t-SNE consistently achieves better recall than PCA, JL, and UMAP, suggesting its effectiveness in preserving neighborhood structure in low dimensions. All dimensionality reduction methods were applied with a target dimension of 3. These results highlight t-SNE's ability to solve approximate nearest neighbor search.

While KNN achieves the best recall performance, it is computationally more expensive to construct compared to other methods. NSW, despite being an approximation of KNN, shows slightly worse recall, trading off quality for efficiency.

\begin{figure}[htbp]
  \centering
  \begin{minipage}[b]{0.49\textwidth}
    \includegraphics[width=\textwidth]{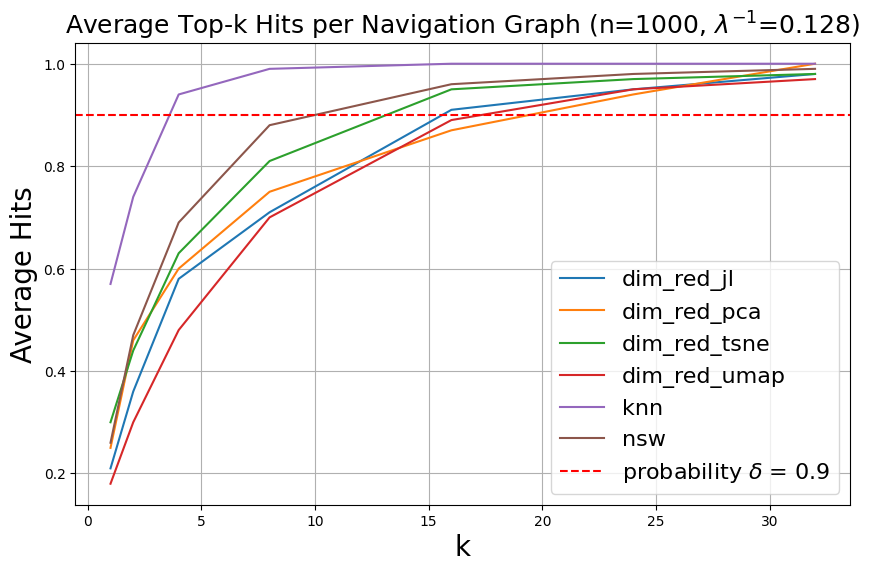}
  \end{minipage}
  \hfill
  \begin{minipage}[b]{0.49\textwidth}
    \includegraphics[width=\textwidth]{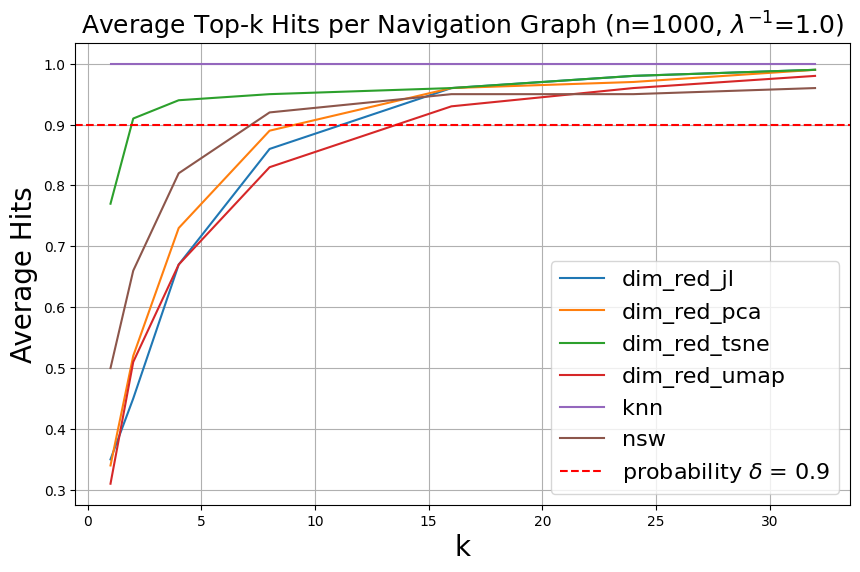}
  \end{minipage}
  \hfill
  \begin{minipage}[b]{0.49\textwidth}
    \includegraphics[width=\textwidth]{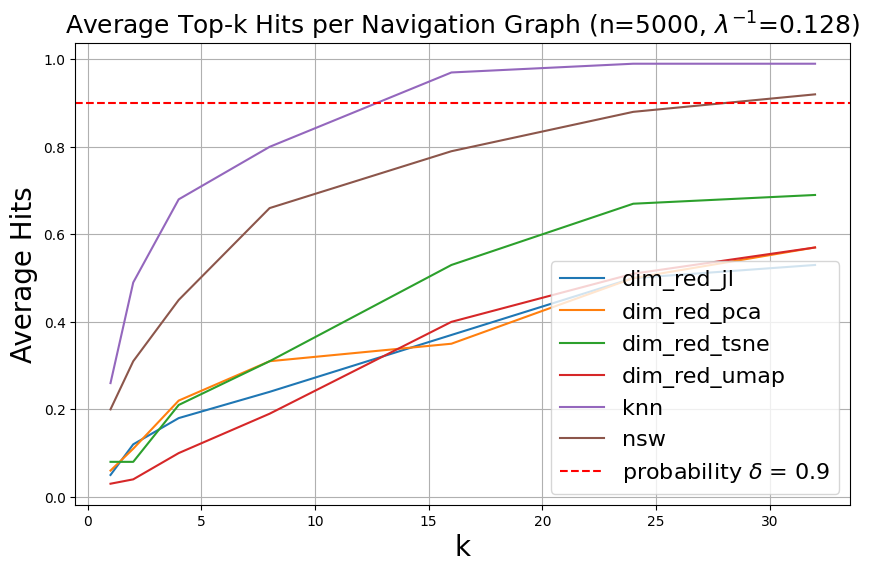}
  \end{minipage}
  \hfill
  \begin{minipage}[b]{0.49\textwidth}
    \includegraphics[width=\textwidth]{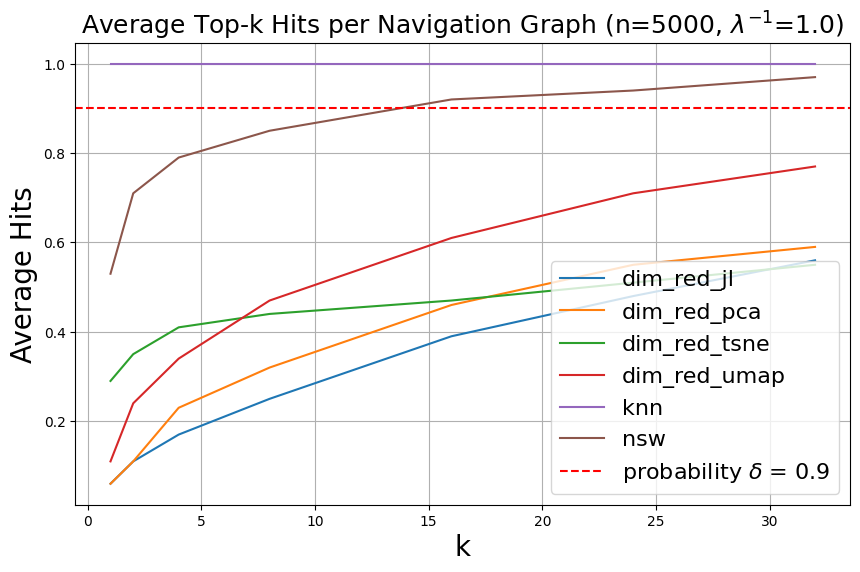}
  \end{minipage}
  \caption{Comparing Recall Bound $\rho_\delta$ for different navigation graphs. \emph{Average Hits} is the fraction of times that at least one point is retrieved within the $k$ nearest points. The red dotted lines show the threshold for $\delta = 0.9$. The intersection of the red dotted line with the plots shows the value of $\rho_{0.9}$ for each navigation graph. The data comes from an exponential distribution with parameter $\lambda$. $\lambda = 1$ is equivalent to uniform distribution and $\frac{1}{\lambda}$ is the mean of this distribution. The dimension of the points here is 32. The values are averaged over 100 runs.}
  \label{fig:recall-bound}
\end{figure}

\paragraph{Angular Distance Metric.} One widely used distance metric in approximate nearest neighbor (ANN) search is the angular (cosine) similarity metric. HENN is compatible with this metric in addition to standard $\ell_p$ norms. For cosine similarity, all vectors are normalized to lie on the unit sphere, and the objective is to find the vector that maximizes the cosine of the angle with a given query vector. This metric is particularly common in applications involving word embeddings, such as GloVe.

To ensure that the theoretical guarantees of HENN extend to the cosine similarity setting, we must show that the associated range space has bounded VC-dimension, similar to the case of $\ell_p$ norms. As discussed in Section~\ref{sec:henn}, ring ranges in HENN are defined as the difference between two balls. Under cosine similarity, the analogous range corresponds to the intersection of the unit sphere with two parallel hyperplanes, representing all points whose cosine similarity with a query point $q$ lies between two thresholds $l$ and $u$. This geometric region is equivalent to a spherical stripe (or band), and it is known that the range space defined by stripes of parallel hyperplanes has VC-dimension $\Theta(d)$. Therefore, the range space under cosine similarity also has a bounded VC-dimension, and the theoretical guarantees of HENN remain valid in this setting.

\paragraph{HENN with Navigation Graphs.} Figure~\ref{fig:henn-nav-graphs} presents the performance of HENN when integrated with three types of navigation graphs: exact KNN, NSW (a variant similar to HNSW), and dimensionality reduction via t-SNE. Among these, the KNN-based integration consistently achieves the highest Recall@$k$ across various values of $k$. HENN+NSW resembles the HNSW implementation, but in this case, we are building $\eps$-nets for each layer instead of any arbitrary random samples.

Interestingly, for real-world datasets with inherent cluster structure, such as FashionMNIST, dimensionality reduction using t-SNE outperforms NSW, likely due to its ability to preserve local neighborhoods. In contrast, under adversarial settings such as data drawn from an exponential distribution, NSW performs more robustly. These results also demonstrate that HENN is compatible with angular distance metrics, as both GloVe and FashionMNIST are evaluated under cosine similarity in this experiment.

\begin{figure}[htbp]
  \centering
  \begin{minipage}[b]{0.49\textwidth}
    \includegraphics[width=\textwidth]{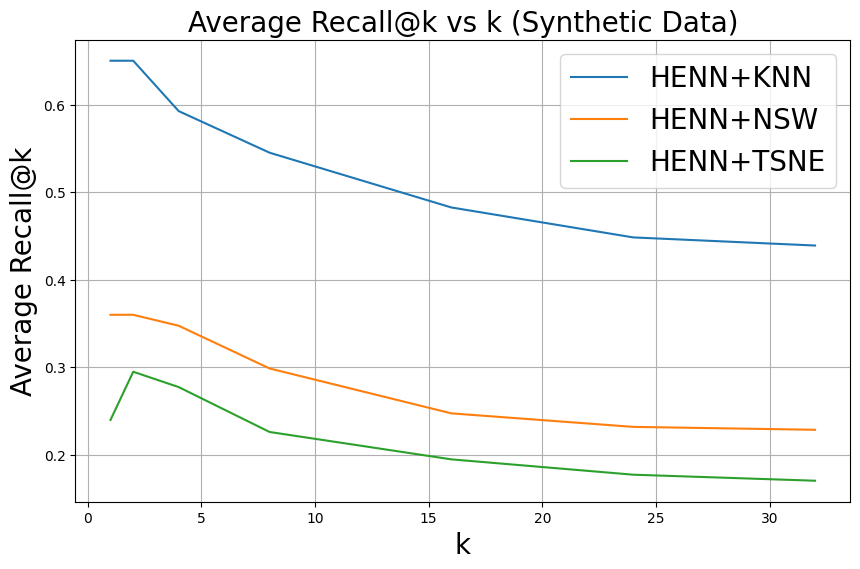}
    \caption*{Exponential ($\lambda^{-1} = 0.128$)}
  \end{minipage}
  \hfill
  \begin{minipage}[b]{0.49\textwidth}
    \includegraphics[width=\textwidth]{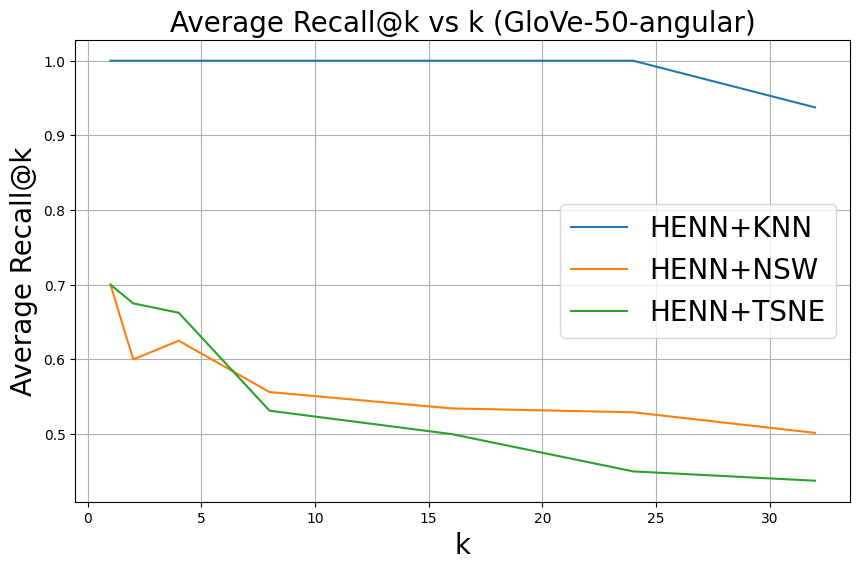}
    \caption*{GloVe ($d = 50$)}
  \end{minipage}
  
  \begin{minipage}[b]{0.49\textwidth}
    \includegraphics[width=\textwidth]{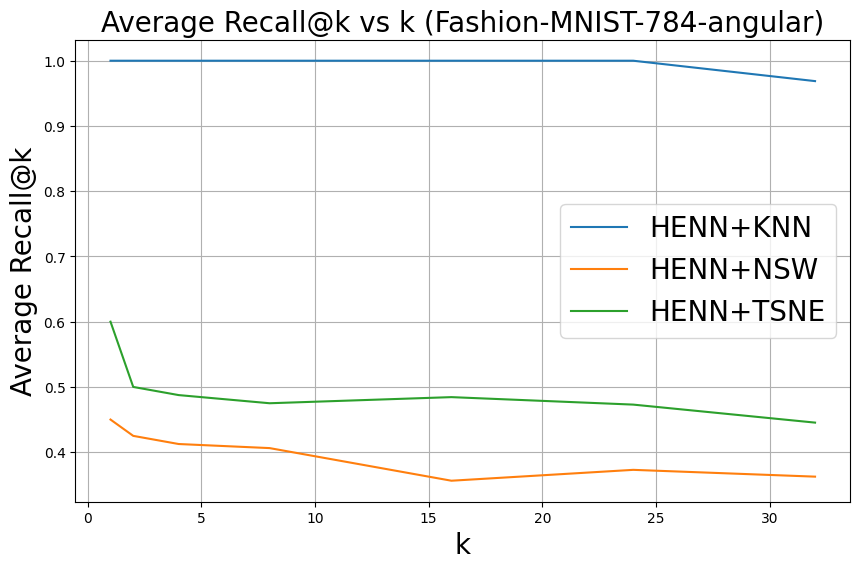}
    \caption*{Fashion-MNIST ($d = 784$)}
  \end{minipage}
  \caption{Comparing the Recall@k for different values of k on HENN structure integrated with different navigation graphs. The size of datasets, $n$, is equal to 10000 in these experiments. The distance metric on the synthetic data is the \underline{Euclidean $\ell_2$-norm}, while for both GloVe and Fashion-MNIST, it is \underline{angular cosine similarity}.}
  \label{fig:henn-nav-graphs}
\end{figure}

\end{document}